\documentclass[11pt]{article}
\usepackage{fullpage}

\usepackage{graphics}
\usepackage[dvips]{epsfig}

\usepackage{amsmath}
\usepackage{amssymb}
\usepackage{amsfonts}
\usepackage{graphicx}

\usepackage[small,compact]{titlesec}
\usepackage{times}

\newtheorem{theorem}{Theorem}[section]
\newtheorem{lemma}{Lemma}[section]
\newtheorem{corollary}{Corollary}[section]

\newtheorem{proposition}{Proposition}[section]

\newcommand{\finclaim}{\hfill $\bullet$}
\newcommand{\qed}{\hfill $\Box$ \bigbreak}
\newenvironment{proof}{\noindent {\bf Proof.}}{\qed}

\newcommand{\cC}{{\cal C}}

\newcommand{\cH}{{\cal H}}
\newcommand{\cS}{{\cal S}}

\newcommand{\cA}{{\cal A}}

\newcommand{\cV}{{\cal V}}
\newcommand{\cM}{{\cal M}}

\newcommand{\E}{{\cal{E}}}

\newcommand{\remove}[1]{}

\begin{document}

\baselineskip  0.18in %  0.2in %0.18in si on veut compact
\parskip     0.0in %    0.1in % 0.0in  pour compacter
\parindent   0.3in %    0.0in % 0.3in pour voir les paragraphes

\title{{\bf Anonymous Meeting in Networks\thanks{
A preliminary version of this paper appeared in the Proc. 24th Annual ACM-SIAM Symposium on Discrete Algorithms (SODA 2013). 
A part of this research was done during the Dieudonn\'e's stay at the Research Chair in Distributed Computing of
the Universit\'{e} du Qu\'{e}bec en Outaouais as a postdoctoral fellow. Supported in part by NSERC discovery grant 
and by the Research Chair in Distributed Computing of
the Universit\'{e} du Qu\'{e}bec en Outaouais.}}}

\author{
Yoann Dieudonn\'{e}\thanks{MIS, Universit\'{e} de Picardie  Jules  Verne Amiens,  France and D\'{e}partement d'informatique, Universit\'{e} du Qu\'{e}bec en Outaouais,
Gatineau, Qu\'{e}bec J8X 3X7,
Canada. E-mail:  yoann.dieudonne@u-picardie.fr. 
}
\and
Andrzej Pelc\thanks{D\'{e}partement d'informatique, Universit\'{e} du Qu\'{e}bec en Outaouais,
Gatineau, Qu\'{e}bec J8X 3X7,
Canada. E-mail: pelc@uqo.ca.}
}

\date{ }
\maketitle

\begin{abstract}
A team consisting of an unknown number of mobile agents, starting from different nodes of an unknown network,
possibly at different times, have to meet at the same node.
Agents are anonymous (identical), execute the same deterministic algorithm and move in synchronous rounds along links of the network. An initial
configuration of agents is called {\em gatherable} if there exists a deterministic algorithm (even dedicated to this particular configuration) that achieves meeting of all
agents in one node. Which configurations are gatherable and how to gather all of them deterministically by the same algorithm?

We give a complete solution of this gathering problem in arbitrary networks. We characterize all gatherable configurations and give two {\em universal} deterministic 
gathering algorithms, i.e., algorithms that gather all gatherable configurations. The first algorithm works under the assumption that a common upper bound $N$
on the size of the network is known to all agents. In this case our algorithm guarantees {\em gathering with detection}, i.e., the existence of a round 
for any gatherable configuration, such that  all agents are at the
same node and all declare that gathering is accomplished. If no upper bound on the size of the network is known, we show that a universal algorithm for gathering
with detection does not exist. Hence, for this harder scenario, we construct a second universal gathering algorithm, which guarantees that, for any gatherable
configuration, all agents eventually get to one node and stop, although they cannot tell if gathering is over. The time of the first algorithm is polynomial in the
upper bound $N$ on the size of the network, and the time of the second algorithm is polynomial in the (unknown) size itself.

Our results have an important consequence for the leader election
problem for anonymous agents in arbitrary graphs.
Leader election is a fundamental symmetry breaking problem in distributed
computing. Its goal is to assign, in some common round, value 1 (leader) to one of the entities and value 0 (non-leader)
to all others.  For anonymous agents in graphs, leader election turns out to be equivalent to
gathering with detection. Hence, as a by-product, we obtain a complete solution of 
the leader election
problem for anonymous agents in arbitrary graphs.

\vspace{2ex}

\noindent {\bf Keywords:} gathering, deterministic algorithm, anonymous mobile agent. 
\end{abstract}

\vfill

\vfill

\thispagestyle{empty}
\setcounter{page}{0}
\pagebreak

%%%%%%%%%%%%%%%%%%%%%%%%%%%%%%%%%%%%%%%%%%%%%%%%%%%%%%%%%%%
\section{Introduction}
%%%%%%%%%%%%%%%%%%%%%%%%%%%%%%%%%%%%%%%%%%%%%%%%%%%%%%%%%%%

\noindent
{\bf The background.}
A team of at least two mobile agents, starting from different nodes of a network, possibly at different times, have to meet at the same node.
This basic task,  known  as {\em gathering} or {\em rendezvous}, has been thoroughly studied in the literature.
This task has even applications in everyday life, e.g., when agents are people that have to meet in a city whose streets form a network.
In computer science,
mobile agents usually represent software agents in computer networks, or mobile robots, if the network is a labyrinth.
The reason to meet may be to exchange data previously collected by the agents,
or to coordinate some future task, such as network maintenance or finding a map of the network.

\noindent
{\bf The model and the problem.}
The network is modeled as an undirected connected graph, referred to hereafter as a graph. 
We seek gathering algorithms that do not
rely on the knowledge of node labels, and can work in anonymous graphs as well  (cf. \cite{alpern02b}). 
The importance of designing such algorithms
is motivated by the fact that, even when nodes are equipped with distinct labels, agents may be unable to perceive them
because of limited sensory capabilities, 
or nodes may refuse to reveal their labels, e.g., due to security or privacy reasons.
Note that if nodes had distinct labels, then agents might explore the graph and meet in the smallest node, hence gathering would reduce to exploration.
On the other hand, we assume that
edges incident to a node $v$ have distinct labels in 
$\{0,\dots,d-1\}$, where $d$ is the degree of $v$. Thus every undirected
edge $\{u,v\}$ has two labels, which are called its {\em port numbers} at $u$
and at $v$. Port numbering is {\em local}, i.e., there is no relation between
port numbers at $u$ and at $v$. Note that in the absence of port numbers, edges incident to a node
would be undistinguishable for agents and thus gathering would be often impossible, 
as the adversary could prevent an agent from taking some edge incident to the current node.

There are at least two agents that start from different nodes of the graph  and  traverse its edges in synchronous rounds.
They cannot mark visited nodes or traversed edges in any way.
The adversary wakes up some of the agents at possibly different times. A dormant agent, not woken up by the adversary,  is woken up by the first agent that visits
its starting node, if such an agent exists. Agents are anonymous (identical) and they execute the same deterministic algorithm. (Algorithms assuming anonymity of agents may be also used if agents refuse to reveal their identities.)
Every agent starts executing the algorithm in the round of its wake-up.
Agents do not  know the topology of the graph or the size of the team. We consider two scenarios: one when agents know
an upper bound on the size of the graph and another when no bound is known. 
In every round an agent may perform some local computations and move to an adjacent node by a chosen port, or stay at the current node.
When an agent enters a node, it learns its degree and the port of entry. When several agents are at the same node in the same round,
they  can exchange all information they currently have. However, agents that cross each other
on an edge, traversing it simultaneously in different directions, do not notice this fact. Also, agents cannot observe any part of the network except the currently visited node.
We assume that the memory of the agents is unlimited: they can be viewed as 
Turing machines.

An initial configuration of agents, i.e., their placement at some nodes of the graph, is called {\em gatherable} if there exists a deterministic algorithm 
(even only dedicated to this particular configuration) that achieves meeting of all
agents in one node, regardless of the times at which some of the agents are woken up by the adversary. In this paper we study the following gathering problem:
\begin{quote}
Which initial configurations are gatherable and how to gather all of them deterministically by the same algorithm?
\end{quote}
In other words, we want to decide which initial configurations are possible to gather, even by an algorithm specifically designed for this particular
configuration, and we want to find a {\em universal} gathering algorithm that gathers all such configurations. We are interested only in {\em terminating}
algorithms, in which every agent eventually stops forever.

\noindent
{\bf Our results.}
We give a complete solution of the gathering problem in arbitrary networks. We characterize all gatherable configurations and give two {\em universal} deterministic 
gathering algorithms, i.e., algorithms that gather all gatherable configurations. The first algorithm works under the assumption that a common upper bound $N$
on the size of the network is known to all agents. In this case our algorithm guarantees {\em gathering with detection}, i.e., the existence of a round 
for any gatherable configuration, such that  all agents are at the
same node and all declare that gathering is accomplished. If no upper bound on the size of the network is known, we show that a universal algorithm for gathering
with detection does not exist. Hence, for this harder scenario, we construct a second universal gathering algorithm, which guarantees that, for any gatherable
configuration, all agents eventually get to one node and stop, although they cannot tell if gathering is over. The time of the first algorithm is polynomial in the
upper bound $N$ on the size of the network, and the time of the second algorithm is polynomial in the (unknown) size itself.

While gathering two anonymous agents is a relatively easy task (cf. \cite{CKP}),
our problem of gathering an unknown team of anonymous agents  
presents the following major difficulty. The asymmetry of the initial configuration
because of which gathering is feasible, may be caused not only by non-similar locations
of the agents in the graph, but by their different situation {\em with respect to other agents}.
Hence a new algorithmic idea is needed in order to gather: agents
that were initially identical, must make decisions based on the memories
of other agents met to date, in order to distinguish their future behavior. 
In the beginning the memory of each agent is a blank slate and
in the execution of the algorithm it records what the agent has seen in 
previous steps of the navigation and what it heard from other agents during meetings.
Even a slight asymmetry occurring  in a remote part of the graph
must eventually influence the behavior of initially distant agents.
Notice that agents in different initial situations may be unaware of this difference in early meetings, as the difference
may be revealed only later on, after meeting other agents. Hence, for example,  an agent may mistakenly "think" that two different agents
that it met in different stages of the algorithm execution, are the same agent. 
Confusions due to this possibility are a significant
challenge absent both in gathering two (even anonymous) agents and in gathering many labeled agents. 

Our results have an important consequence for the leader election
problem for anonymous agents in arbitrary graphs.
Leader election \cite{Ly} is a fundamental symmetry breaking problem in distributed
computing. Its goal is to assign, in some common round, value 1 (leader) to one of the entities and value 0 (non-leader)
to all others.  For anonymous agents in graphs, leader election turns out to be equivalent to
gathering with detection (see Section 5). Hence, 
as a by-product, we obtain a complete solution of 
the leader election
problem for anonymous agents in arbitrary graphs.

%--------------------------------------------------

\noindent
{\bf Related work.}
Gathering was mostly studied for two mobile agents and in this case it is usually called rendezvous.
An extensive survey of  randomized rendezvous in various scenarios  can be found in
\cite{alpern02b}, cf. also  \cite{alpern95a,alpern02a,anderson90,israeli}. 
Deterministic rendezvous in networks was surveyed in \cite{Pe}.
Several authors
considered the geometric scenario (rendezvous in an interval of the real line, see, e.g.,  \cite{baston01,gal99},
or in the plane, see, e.g., \cite{anderson98a,anderson98b}).
Gathering more than two agents was studied, e.g., 
in \cite{israeli,lim96}. In~\cite{YY} the authors considered 
rendezvous of many agents with unique labels, and gathering many labeled agents in the presence of Byzantine agents was studied in \cite{DPP}. 
The problem was also studied in the context of multiple robot systems, cf.
\cite{CP05,fpsw}, and fault tolerant gathering of robots in the plane was studied, e.g., in \cite{AP06,CP08}. 

For the deterministic setting a lot of effort was dedicated to the study of the feasibility of rendezvous, and to the time required to achieve this task, when feasible. For instance, deterministic rendezvous with agents equipped with tokens used to mark nodes was considered, e.g., in~\cite{KKSS}. Deterministic rendezvous of two agents that cannot mark nodes but have unique labels was discussed in \cite{DFKP,KM,TSZ14}.
These papers were concerned with the time of rendezvous in arbitrary
graphs. In \cite{DFKP} the authors showed a rendezvous algorithm polynomial in the size of the graph, in the length of the shorter
label and in the delay between the starting time of the agents. In \cite{KM,TSZ14} rendezvous time is polynomial in the first two of these parameters and independent of the delay.

Memory required by two anonymous agents to achieve deterministic rendezvous was studied in \cite{FP,FP2} for trees and in  \cite{CKP} for general graphs.
Memory needed for randomized rendezvous in the ring was discussed, e.g., in~\cite{KKPM11}. 

Apart from the synchronous model used, e.g., in \cite{CKP,DFKP,TSZ14,YY} and in this paper, several authors investigated asynchronous rendezvous in the plane \cite{CFPS,fpsw} and in network environments
\cite{BCGIL,CLP,DGKKP,DPV,GP}.
In the latter scenario the agent chooses the edge which it decides to traverse but the adversary controls the speed of the agent. Under this assumption rendezvous
in a node cannot be guaranteed even in very simple graphs and hence the rendezvous requirement is relaxed to permit the agents to meet inside an edge.
In particular, in \cite{GP} the authors studied asynchronous rendezvous of two anonymous agents. 

Gathering many anonymous agents in a ring  \cite{DDN,KKN,KMP} and in a tree \cite{FIPS} was studied in a different model, where agents 
execute Look-Compute-Move cycles in an asynchronous way and can see positions of other agents during the Look operation. While agents have the advantage
of periodically seeing other agents, they do not have any memory of past events, unlike in our model.

Tree exploration by many agents starting at the same node was studied in \cite{FGKP}.
Some of the techniques of this paper have been later used in our paper \cite{DP} in the context of anonymous graph exploration by many agents.

The leader election problem was introduced in \cite{LL} and has been mostly studied in the scenario
where all nodes have distinct labels and the leader is found among them \cite{FL,HS,P}. 
In \cite{HKMMJ}, the leader election problem was
approached in a model based on mobile agents for networks with labeled nodes.
Many authors \cite{An,AtSn,BV,YK2,YK3} studied leader election
in anonymous networks. In particular, \cite{YK3} characterized message-passing networks in which
leader election among nodes can be achieved when nodes are anonymous.
%They assume that the number of nodes of the network is known to all nodes. 
%Characterizations of feasible instances for leader election were provided in~\cite{C,CM}.
Memory needed for leader election in unlabeled networks was studied in \cite{FuPe}.

\section{Preliminaries}\label{prelim}

Throughout the paper, 
%if $s$ is an integer, $|s|$ denotes the length of its binary representation, and if $S$ is a set, $|S|$ denotes its number of elements. 
the number of nodes of a graph is called its size.
In this section we recall four procedures known from the literature, that will be used as building blocks in our algorithms. 
The aim of the first two procedures is graph exploration, i.e., visiting all nodes of the graph by a single agent (cf., e.g., \cite{CDK,Re}). 

{The first of these procedures assumes an upper bound $N$ on the size of the graph. It is based on universal exploration sequences (UXS) and is a corollary of the result of Reingold \cite{Re}. It allows the agent to traverse all nodes of any graph of size at most $N$, starting from any node of this graph, using $P(N)$ edge traversals, where $P$ is some polynomial. After entering a node of degree $d$ by some port $p$,
the agent can compute the port $q$ by which it has to exit; more precisely $q=(p+x_i)\mod d$, where $x_i$ is the corresponding term of the UXS of length $P(N)$. 

The second procedure makes no assumptions on the size but 
it is performed by an agent using a fixed token placed at the starting node of the agent.}
(It is well known that a terminating exploration even of all anonymous rings of unknown size by a single agent without a token is impossible.)
In our applications the roles of the token and of the exploring agent will be played by agents or by groups of agents.

The first procedure works in time polynomial in the known upper bound $N$ on the size of the graph and the second in time polynomial in the size of the graph. At the end of each procedure all nodes of the graph are visited.
Moreover, at the end of the second procedure the agent is with the token and has a complete map of the graph with all port numbers marked.
We call the first procedure $EXPLO(N)$ and the second procedure $EST$, for {\em exploration with a stationary token}.
{We denote by $T(EXPLO(N))$ the maximum time of execution of the procedure  $EXPLO(N)$
in a graph of size at most $N$.}

%The first procedure, based on universal exploration sequences,
%is due to Reingold \cite{Re}. It works for a graph with known upper bound $n$ on its size 
%and visits all nodes of the graph, starting from an arbitrary unknown node,
%in time polynomial in $n$. We denote this procedure by $RP_n$ and the respective polynomial by $p(n)$.

Before describing the third procedure we define the
following notion from \cite{YK3}. Let $G$ be a graph and $v$ a node of $G$.  
The {\em view} from $v$ is the infinite rooted tree $\cV(v)$ with labeled ports, defined as follows.
Nodes of this tree are all (not necessarily simple) finite paths in $G$ starting at $v$ and represented as the corresponding sequences of port numbers.
The root of the tree is the empty path, and children of path $\pi$ of length $x$ are all paths of length $x+2$, whose prefix is $\pi$.
A child $\rho$ of node $\pi$, such that $\rho=\pi pq$ (with juxtaposition standing for concatenation) is connected with node $\pi$ by an edge
that has port $p$ at node $\pi$ and port $q$ at node $\rho$.

The {\em truncated view} $\cV^l(v)$ at depth $l$ is the rooted subtree of  $\cV(v)$ induced by all nodes at distance  at most $l$ from the root.

The third procedure, described in \cite{CKP},  permits a single anonymous agent starting at node $v$ of an $n$-node graph, to find a  positive integer $S(v) \in \{1,\dots , n\}$,
called the signature of the agent, such that 
$\cV(v)= \cV(w)$ if and only if $S(v)=S(w)$. This procedure, called $SIGN(N)$ works for any graph $G$ of known upper bound $N$ on its size and 
its running time is polynomial in $N$. After the completion of $SIGN(N)$, {graph $G$ has been completely explored (i.e., all the nodes of $G$ have been visited at least once) and} the agent is back at its starting node.
We denote by $T(SIGN(N))$ the maximum time of execution of the procedure $SIGN(N)$ in a graph of size at most $N$.

Finally, the fourth procedure is for gathering two agents in a graph of unknown size. It is due to Ta-Shma and Zwick \cite{TSZ14} and relies on the
fact that agents have distinct labels. (Using it as a building block in our scenario of anonymous agents is one of the difficulties that we need to overcome.)
Each agent knows its own label (which is a parameter of the algorithm) but not the label of the other agent.
We will call this procedure $TZ(\ell)$, where $\ell$ is the label of the executing agent. In  \cite{TSZ14}, the authors give a polynomial $P$
in two variables, increasing in each of the variables,
such that, if there are agents with distinct labels $\ell_1$ and $\ell_2$ operating in a graph of size $n$, appearing at their starting positions in possibly
different times, then they will meet after at most $P(n,|\ell|)$ rounds since the appearance of the later agent, where $\ell$ is the smaller label.
Also, if an agent with label $\ell_i$  performs $TZ(\ell_i)$ for $P(n,|\ell_i|)$ rounds and the other agent is inert during this time, the meeting is guaranteed.

We will also use a notion similar to that of the view but reflecting the positions of agents in an initial configuration.
Consider a graph $G$ and an initial configuration of agents in this graph. Let $v$ be a node occupied by an agent. The {\em enhanced view} from $v$ is the couple 
 $(\cV(v),f)$, where $f$ is a binary valued function defined on the set of nodes of $\cV(v)$, such that $f(\pi)=1$ if there is an agent at the unique node $w$ in $G$
 that is reached from $v$ by path $\pi$, and $f(\pi)=0$ otherwise. (Recall that nodes of $\cV(v)$ are paths in $G$, represented as sequences of port numbers.)
 Thus the enhanced view of an agent additionally marks in its view the positions of other agents in the initial configuration.
 
 An important notion used throughout the paper is the {\em memory} of an agent. 
 Intuitively, the memory of an agent in a given round is the total information the agent collected since its wake-up,
 both by navigating in the graph and by exchanging information with other agents met until this round. This is formalized as follows. The memory of an agent $A$ at
 the end of some round which is  the $t$-th round since its wake-up, is the sequence $(M_0,M_1,\dots ,M_t)$, where $M_0$ is the information of the agent at its start
 and $M_i$ is the information acquired in round $i$. The terms $M_i$ are defined as follows. Suppose that in round $i$ the agent $A$ met agents $A_1,\dots ,A_k$ at node $v$ of degree $d$. Suppose that agent $A$ entered node $v$ in round $i$ leaving an adjacent node $w$ by port $p$ and entering $v$ by port $q$. 
 Suppose that agent $A_j$ entered node $v$ in round $i$ leaving an adjacent node $w_j$ by port $p_j$ and entering $v$ by port $q_j$. If some agent did not move in
 round $i$ then the respective ports are replaced by $-1$. The term $M_i$ for agent $A$, called its $i$-th {\em memory box} is defined as the sequence
 $(d,p,q,\{(p_1,q_1,R_1),\dots, (p_k,q_k,R_k)\})$, where $R_j$ is the memory of the agent $A_j$ at the end of round $i-1$. This definition is recursive with respect to 
  global time (unknown to agents), starting at the wake-up of the earliest agent by the adversary. Note that if an agent is woken up by the adversary 
  at a node of degree $d$ and no other agents are at this node at this time then {$M_0=(d,-1,-1,\emptyset)$} for this agent (at its wake-up the agent sees only the degree of its
  starting position). If an agent is woken up by some other agents, it also learns their memories to date. Note also that, since the memory of the agent is the total  
  information it acquired to date, any deterministic algorithm used by agents to navigate (including any deterministic gathering algorithm) may be viewed as a function
  from the set of all memories into  the set of integers greater or equal to $-1$ telling the agent to take a given port $p$ 
  (or stay idle, in which case the output is $-1$), if it has a given memory. It should be noted that memories of agents may often be (at least) exponential in the size of the graph in which they navigate.
  
  If two agents meet, they must necessarily have different memories. Indeed, if they had identical memories, then they would be woken up in the same round
  and they would traverse identical paths to the meeting node. This contradicts the assumption that agents start at different nodes.

 {We will end this section by introducing the order $\prec$ on the set $\cM$ of all possible memories, and the notion of $Pref_t(\cM)$ that will be used throughout the paper. Let $<_\alpha$ be any linear order on the set of all memory boxes (one such simple order is to code all possible memory boxes as binary sequences in some canonical way and use lexicographic order on binary sequences). Let $<_\beta$ be a lexicographic order on the set $\cM$ based on the order $<_\alpha$. Now the linear order $\prec$ on the set $\cM$ is defined as follows: $\cM_1 \prec \cM_2$ if (1) the number of memory boxes of $\cM_1$ is strictly smaller than that of $\cM_2$ or (2) $\cM_1$ and $\cM_2$ have the same number of memory boxes and $\cM_1 <_\beta \cM_2$. When $\cM_1 \prec \cM_2$ (resp. $\cM_2 \prec \cM_1$), we say memory $\cM_2$ is {\em larger} (resp. {\em smaller}) than $\cM_1$. The order $\prec$ has the property that if the memory of agent $A$ is smaller than the memory of agent $B$ in some round then it will remain smaller in all subsequent rounds. Let $A$ be an agent with a memory $\cM=(M_0,M_1,\dots ,M_s)$ in round $t'$ (with $s\leq t'$): the memory of agent $A$ in round $t$ for $t\leq t'$ is denoted $Pref_t(\cM)$ and is equal to $(M_0,M_1,\dots ,M_{s-(t'-t))}$ if $t'-t<s$, and to the empty word otherwise.}

\section{Known upper bound on the size of the graph}

In this section we assume that a common upper bound $N$ on the size of the graph is known to all agents at the beginning.
The aim of this section is to characterize gatherable configurations and give a {\em universal} gathering algorithm that gathers {\em with detection} all gatherable configurations. The time of such an algorithm is the number of rounds between the wake-up of the first agent and the round when
all agents declare that gathering is accomplished.
Consider the following condition on an initial configuration in an arbitrary graph.
\begin{itemize}
\item
[{\bf G:}] There exist agents with different views, and each agent has a unique enhanced view.
\end{itemize}

We will prove the following result.

\begin{theorem}\label{eq}
An initial configuration is gatherable if and only if it satisfies the condition {\bf G}. If a common upper bound $N$ on the size of the graph is known to all agents then
there exists an algorithm for gathering with detection all gatherable configurations. This algorithm works in time polynomial in $N$.
\end{theorem}

In order to appreciate the full strength of Theorem \ref{eq} notice that it can be rephrased as follows.  If  an initial configuration does not satisfy
condition {\bf G} then there is no algorithm (even no algorithm dedicated to this specific configuration, knowing it entirely) that permits to gather this configuration, even gather it without detection: simply no algorithm can bring all agents simultaneously to one node. On the other hand, assuming that all agents know a common 
upper bound on the size of the graph, there is a {\em universal} algorithm that gathers {\em with detection} all initial configurations
satisfying condition {\bf G}. Our algorithm works in time polynomial in any commonly known upper bound $N$ on the size of the graph. Hence,  if this upper bound is
polynomial in the size of the graph, the algorithm is polynomial in the size as well. 
In the next section we will show how the positive part of the result changes when no upper bound on the size of the graph is known. 

It is also important to stress the global nature of the second part of condition {\bf G}. While the first part concerns views of individual agents, which depend only on
the position of the agent in the graph, the second part concerns {\em enhanced} views, which depend, for each agent, on the positions of {\em all} agents in the graph.
In our algorithms, the actions of each agent eventually depend on its enhanced view, and hence require learning the initial positions of other agents with respect to the initial position of the given agent.
(Since not only agents but also nodes are anonymous, no ``absolute'' positions can ever be learned).
 This can be done
only by meeting other agents and learning these positions, sometimes indirectly, with agent $A$ acting as intermediary in conveying the initial position of agent $B$ 
to agent $C$. In short, the execution of a gathering algorithm for an agent depends on the actions of other agents. This is in sharp contrast to rendezvous of two anonymous agents, cf., e.g., \cite{CKP}, where actions of each agent depend only on its  view, whose sufficient part (a truncated view) can be constructed
individually by each agent without any help from the outside. This is also in sharp contrast with gathering or other tasks performed by two or more {\em labeled} agents, when actions of each agent depend on its label known by the agent in advance, and simple decisions based on comparison of labels are made after meetings
of subsets of agents when the task involves more than two agents,  cf., e.g., \cite{DPV}. 

The rest of the section is devoted to the proof of Theorem \ref{eq}.
We start with the following lemma.

\begin{lemma}\label{not}
If an initial configuration does not satisfy condition {\bf G}, then it is not gatherable.
\end{lemma}

\begin{proof}
Suppose that an initial configuration $C$ does not satisfy condition {\bf G} and that agents execute the same deterministic algorithm.
First suppose that the configuration $C$ does not satisfy the first part of condition  {\bf G}, i.e., that the views of all agents are identical.
Suppose that the adversary wakes up all agents in the same round. We show that no pair of agents can meet. Suppose, for contradiction, that 
agents $A_1$ and $A_2$ are the first to meet and that this occurs in round $t$ from the common start. Since the initial views of the agents are identical
and they execute the same deterministic algorithm, the sequences of port numbers encountered by both agents are identical and hence they both
enter the node at which they first meet by the same port. This is a contradiction.

Now suppose that the configuration $C$ does not satisfy the second part of condition  {\bf G}, i.e., that there exists an agent whose enhanced view is not unique. 
Consider distinct agents
$A$ and $A'$ that have identical enhanced views. Let $B$ be any agent and suppose that  $q$ is a sequence of port numbers that leads
from agent $A$ to agent $B$ in the enhanced view of agent $A$. (Notice that there may be many such sequences, corresponding to different paths 
leading from $A$ to $B$.) Then there exists an agent $B'$ such that $q$ is a sequence of port numbers that leads
from agent $A'$ to agent $B'$ in the enhanced view of agent $A'$. Agent $B'$ has the same enhanced view as agent $B$. 
Since agents $A$ and $A'$ were different, agents $B$ and $B'$ are different as well. Hence for every agent there exists another agent whose enhanced view
is identical. Call such agents {\em homologs}.

Suppose again that the adversary wakes up all agents in the same round. We will show that, although now some meetings of agents are possible, homologs will
never meet. Suppose that agents $A$ and $A'$ are homologs. Since $A$ and $A'$ have the same enhanced view at the beginning, it follows by induction
on the round number that their memory will be identical in every round. Indeed, whenever $A$ meets an agent $B$ in some round, its homolog $A'$ meets
a homolog $B'$ of $B$ in the same round and hence the memories of $A$ and $A'$ evolve identically. In particular, since all agents execute the same
deterministic algorithm, agents $A$ and $A'$ follow identical sequences of port numbers on their paths. As before, if they met for the first time at some node, they
would have to enter this node by the same port. This is a contradiction.
 \end{proof}
 
 The other (much more difficult) direction of the equivalence from Theorem \ref{eq} will be shown by constructing an algorithm which, 
 executed by agents starting from any initial
 configuration satisfying condition {\bf G}, accomplishes gathering with detection of all agents. (Our algorithm uses the knowledge of an upper bound $N$ of the size
 of the graph: this is enough to prove the other direction of the equivalence, as for any specific configuration satisfying condition {\bf G} even a gathering
 algorithm dedicated to this specific configuration is enough to show that the configuration is gatherable, and such a dedicated algorithm knows the exact size
 of the graph. Of course, our algorithm accomplishes much more: knowing just an upper bound on the size of the graph it gathers {\em all} configurations satisfying condition {\bf G} and does this {\em with detection}.) We first give a high-level idea of the algorithm, then describe
 it in detail and prove its correctness. From now on we assume that the initial configuration satisfies condition {\bf G}.
 
 \vspace*{0.3cm}
 \noindent
 {\bf Idea of the algorithm.}
 At high level the algorithm works in two stages. The aim of the first stage for any agent is to meet another agent, in order to perform later an exploration 
 in which one of the agents will play the role of the token and the other the role of the explorer (roles will be decided comparing memories of the agents).
 To this end the agent starts an exploration of the graph using the known upper bound $N$ on its size. The aim of this exploration is to wake up all, possibly still dormant
 agents. Afterwards,  in order to meet, 
 agents use the procedure $TZ$ mentioned in Section 2. 
 However, this procedure requires a label  for each agent and our agents are anonymous. One way to differentiate agents
 and give them labels is to find their views: a view (truncated to $N$) from the initial position could serve as a label because the first part of condition {\bf G} guarantees
 that there are at least two distinct views. However, it is {\em a priori} not at all clear how to find the view of an agent (truncated to $N$) in time polynomial in $N$.
 (Recall that the size of the view truncated to $N$ is exponential in $N$.) Hence we use
 procedure $SIGN(N)$ mentioned in Section 2, which is polynomial in $N$ and can still assign labels to agents {after exploring the graph}, producing at least two different labels. 
 Then, performing procedure $TZ$ for a sufficiently long time guarantees a meeting for every agent.
 
 In the second stage each explorer explores the graph using procedure $EXPLO(N)$ and then backtracks to its token
 left at the starting node of the exploration. After the backtrack, memories of the token and of the explorer are updated to check for anomalies caused by other agents
 meeting in the meantime either the token or the explorer. Note that, due to the fact that some agents may have identical memories at this stage of the algorithm, 
 an explorer may sometimes falsely consider another token as its own, due to their identical memories. By contrast, an end of each
 backtrack is a time when an explorer can be sure that it is on its token and the token can be sure that its explorer is with it.  
 
 Explorers repeat these explorations with backtrack again and again, with the aim of creating meetings with other agents and detecting anomalies.
 As a consequence of these anomalies some agents merge with others, mergers being decided on the basis of the memories of the agents.
Each explorer eventually either merges with some token or performs an exploration without anomalies. In the latter case it waits a prescribed amount of time with its token: if
no new agent comes during the waiting time, the end of the gathering is declared, otherwise another exploration is launched. It will be proved that eventually, due to the second part of condition {\bf G}, all agents merge with the same
token $B$ and then, after the last exploration made by the explorer $A$ of $B$ and after undisturbed waiting time, the end of the gathering is correctly declared.

 We now give a detailed description of the algorithm.
 
 \vspace*{0.3cm}
 
 \noindent
 {\bf Algorithm Gathering-with-Detection} with parameter $N$ (upper bound on the size of the graph)
 
  \vspace*{0.2cm}
 
 During the execution of the algorithm an agent can be in one of the following six states: {\tt setup},  {\tt cruiser}, {\tt shadow}, {\tt explorer},  {\tt token}, {\tt searcher},
 depending on its memory.
 For every agent $A$ in state  {\tt shadow} there is exactly one agent $B$ in some state different from  {\tt shadow}, called the {\em guide} of $A$. We will also
 say that $A$ is a shadow of $B$.
 Below we describe the actions of an agent $A$ in each of the states and the transitions between the states. At wake-up agent $A$ enters the state  {\tt setup}.
 
  \vspace*{0.2cm}
 
  \noindent
 {\bf State} {\tt setup}.
 
{Agent $A$ performs $SIGN(N)$ visiting all nodes (and waking up all still dormant agents) and finding the signature of its initial position $v$, called the {\em label} of agent $A$.}
 Agent $A$  transits to state  {\tt cruiser}.
 
  \vspace*{0.2cm}
 
  \noindent
 {\bf State} {\tt cruiser}.
 
 Agent $A$ performs $TZ(\ell)$, where $\ell$ is its label, until meeting an agent in state  {\tt cruiser} or {\tt token} at a node $v$. 
 {When such a meeting occurs, we consider $2$ cases.}

  \noindent
 Case 1.
 Agent $A$  meets an agent $B$ in state {\tt token}.\\
{We consider $2$ subcases}

{Subcase 1.1.}
{Agent $B$ is not with its explorer at the time of the meeting. Then agent $A$ transits to state {\tt shadow} of $B$.}

{Subcase 1.2.}
 {Agent $B$ is with its explorer $C$ at the time of the meeting. Then agent $A$ transits to state {\tt shadow} of $C$.}

  \noindent
 Case 2.
 Agent $A$ does not meet an agent  in state {\tt token}.\\ 
 Then there is at least one other agent in state {\tt cruiser} at node $v$ {(because, as mentionned above, the considered meeting involves an agent in state {\tt cruiser} or {\tt token})}. {We consider $2$ subcases.}
 
 Subcase 2.1.
 Agent $A$  has the largest memory among all agents in state  {\tt cruiser} at node $v$.\\
 Then agent $A$  transits to state {\tt explorer}.
 
 Subcase 2.2.
 Agent $A$  does not have the largest memory among all agents in state  {\tt cruiser} at node $v$.\\
 If there is exactly one agent $B$ in state {\tt cruiser}  with memory larger than $A$ at node $v$, then $A$ transits to state {\tt token}.
 Otherwise, it becomes {\tt shadow} of the agent in state  {\tt cruiser} at node $v$ with largest memory.
 
  \vspace*{0.2cm}
 
 \noindent
 {\bf State}  {\tt shadow}.
 
 Agent $A$ has exactly one guide and is at  the same node as the guide in every round. In every round it makes the same move as the guide.
 If the guide $B$ transits itself to state  {\tt shadow} and gets agent $C$ as its guide, then agent $A$ changes its guide to $C$ as well. 
 Agent $A$ declares that gathering is over if the unique agent in state {\tt explorer} collocated with it makes this declaration.  
 
 Before describing the actions in the three remaining states, we define the notion of {\em seniority} of an agent in state {\tt token} (respectively {\tt explorer}).
The seniority in a given round is the number of rounds from the time when the agent became {\tt token} (respectively {\tt explorer}).

 \vspace*{0.2cm}
 
  \noindent
 {\bf State} {\tt explorer}
 
 When agent $A$ transits to state {\tt explorer}, there is another agent $B$ that transits to state  {\tt token} in the same round at the same node $v$.
 Agent $B$ is called the token of $A$. Agent $A$ has a variable $recent$-$token$ that it initializes to the memory of $B$ in this round. 
 Denote by $EXPLO^*(N)$ the procedure  $EXPLO(N)$ followed by a complete backtrack  in which the agent traverses all edges traversed
in  $EXPLO(N)$ in the reverse order and the reverse direction. The variable $recent$-$token$ is updated in the beginning of each execution of $EXPLO^*(N)$.
 An execution of $EXPLO^*(N)$  is called {\em clean} if 
the following condition is satisfied: in each round  during this execution, in which $A$ met an agent $C$ {that is not in state {\tt shadow}}, the memory of $C$ is equal to that of $B$,
and in each round during this execution, in which the token $B$ was met by an agent $D$, the memory of $D$ was equal to that of $A$. Notice that 
after the execution of  $EXPLO^*(N)$, agent $A$ is together with its token $B$ and thus they can verify if the execution was clean, by inspecting their memories.
The execution time of $EXPLO^*(N)$ is at most $2T(EXPLO(N))$.
%Let $t_1, \dots, t_k$ be such that $A$ met an agent in state {\tt  token} in rounds $t+t_1,\dots t+t_k$ during this execution

 After transiting to state {\tt explorer}, agent $A$ waits for {$T(SIGN(N))+P(N,L)$ rounds}, where $L$ is the largest possible label (it is 
 polynomial in $N$). 
 Then it executes the following protocol:
\newpage
 
  \vspace*{0.2cm}
  \noindent
 {\bf while} $A$ has not declared that gathering is over {\bf do}
 
 \noindent
  \hspace*{0.5cm}{\bf do}\\
   \noindent
 \hspace*{1cm}$EXPLO^*(N)$\\
  \noindent
  \hspace*{1cm}/*now agent $A$ is with its token.*/
 
 \vspace*{0.2cm}
  \noindent
  \hspace*{1cm}{\bf if} {in round $t'$} agent $A$ met an agent $C$ in state {\tt token} of higher\\
   \noindent
  \hspace*{1cm}seniority than that of $A$ or of equal seniority but such
  that \\ 
   \noindent
   \hspace*{1cm}$recent\mbox{-}token\prec Pref_t( \cM_C)$ where $\cM_C$ is the memory of agent $C$ and\\
    \noindent 
    \hspace*{1cm}$t$ is the last round {before $t'$}
    when agent $A$ updated its variable \\
    \noindent 
    \hspace*{1cm}$recent$-$token$ {\bf then} $A$ transits to state {\tt searcher}
    
    \vspace*{0.2cm}
     \noindent
     \hspace*{1cm}{\bf if} $B$ was visited in round $t'$ by an agent $C$ in state {\tt explorer} of higher\\ 
      \noindent 
    \hspace*{1cm}seniority than that of $B$
     or of equal seniority but such that\\
      \noindent 
      \hspace*{1cm}{$Pref_t(\cM_B) \prec R$} where $\cM_B$ is the memory of agent $B$, $t$ is the last \\
       \noindent 
      \hspace*{1cm}{
round before $t'$ when agent $C$ updated its variable $recent$-$token$ and $R$}\\ 
\noindent 
      \hspace*{1cm}{is the variable $recent$-$token$ of agent $C$ in round $t'$} {\bf then} $A$ transits to \\
\noindent 
      \hspace*{1cm}state {\tt searcher}
    
     \vspace*{0.2cm}
     \noindent
       \hspace*{0.5cm}{\bf until} the execution of $EXPLO^*(N)$ is clean\\

     \noindent
       \hspace*{0.5cm}{agent $A$ waits $2 \cdot  T(EXPLO(N))$ rounds: this waiting period is interrupted} \\
       \hspace*{0.5cm}{if $A$ is visited by another agent;} 
       
        \vspace*{0.2cm}
        \noindent
       \hspace*{0.5cm}{{\bf if} the waiting period of $2 \cdot  T(EXPLO(N))$ rounds has expired\\}
       \hspace*{0.5cm}{without any interruption {\bf then} $A$ declares that gathering is over.\\}
\vspace*{0.2cm}
 {{\bf endwhile}}

  \vspace*{0.2cm}
 
  \noindent
{\bf State} {\tt token}
 
  When agent $A$ transits to state {\tt token}, there is another agent $B$ that transits to state  {\tt explorer} in the same round at the same node $v$.
 Agent $B$ is called the explorer of $A$. 
 %Agent $A$ has a variable $recent-explorer$ that it initializes to the memory of $B$ in this round. 
 Agent $A$ remains  idle at a node $v$ and does not change its state, except when its explorer $B$ transits to state {\tt searcher}. In this case it transits
 to state {\tt shadow} and $B$ becomes its guide. 
    Agent $A$ declares that gathering is over if the unique agent in state {\tt explorer} collocated with it makes this declaration.
    
     \vspace*{0.2cm}
 
  \noindent
  {\bf State} {\tt searcher}
    
{Agent $A$ performs an entire execution of $EXPLO^*(N)$ until its termination, regardless of any meetings it could make during this execution. Then the agent starts another execution of $EXPLO^*(N)$ which is stopped as soon as agent $A$ meets an agent $B$ in state {\tt token}.}
{If at the time of the meeting agent $B$ is not with its explorer $C$ then agent $A$ transits to state {\tt shadow} of $B$ ($B$ becomes its guide). Otherwise, agent $A$ transits to state {\tt shadow} of $C$ ($C$ becomes its guide).}

\vspace*{0.5cm}
The proof of the correctness of the algorithm is split into the following lemmas.

\begin{lemma}\label{term}
In Algorithm Gathering-with-Detection every agent eventually stops after time polynomial in $N$ and declares that gathering is over.
\end{lemma}

\begin{proof}
At its wake-up an agent $A$ enters state  {\tt setup} and remains in it for at most {$T(SIGN(N))$ rounds} (the time to complete an exploration and find
its label $\ell$) and then transits to state {\tt cruiser}. We will prove that
in state {\tt cruiser} agent $A$ can spend at most  {$T(SIGN(N))+2P(N,\ell)$ rounds}.
We will use the following claim.

\vspace*{0.3cm}
\noindent
{\bf Claim 1.} Let $t$ be the first round{, if any,} in which an agent transits to state {\tt token}. Then there exists an agent $B$ that remains in state {\tt token}
and is idle from round $t$ on.

{To prove the claim, let $Z$ be the set of agents that transited to state {\tt token} in round $t$.  In every round $t'>t$, {each} agent from $Z$ with the current largest memory
remains in state {\tt token} and stays idle. Indeed, the reasons why such an agent, call it $X$, could leave the state token in round $t'$ all lead to a contradiction. There are four such reasons.}

\begin{itemize}
\item{ {Case~1.} The token $X$ was visited by an agent in state {\tt explorer} of higher seniority. We get a contradiction with the fact that the agents belonging to $Z$ have the highest seniority.}

\item{ {Case~2.} The explorer of the token $X$ met an agent $Y$ in state {\tt token} of higher seniority. Since the explorer of $X$ has the same seniority as $X$, by transitivity the seniority of agent $Y$ is higher than that of agent $X$ which is a contradiction with the definition of $Z$.}

\item{ {Case~3.} In round $k< t'$ the token $X$ was visited by an agent $Y$ in state {\tt explorer} of equal seniority but such that $Pref_s(\cM_X) \prec R$ where $\cM_X$ is the memory of agent $X$, $s$ is the last round before $k$ when agent $Y$ updated its variable $recent$-$token$ and $R$ is the variable $recent$-$token$ of agent $Y$ in round $k$. This case is impossible. Indeed, by definition, agent $X$ has one of the highest memories among the agents from $Z$ in round $t'$. Hence, according to the definition of order $\prec$ given in Section~\ref{prelim}, agent $X$ had one of the highest memories among the agents from $Z$ in all rounds between $t$ and $t'$. Moreover, since $Y$ and $X$ have the same seniority, this implies that the token of $Y$ belongs to $Z$. Hence, in round $s$ the memory of agent $X$ is greater than or equal to the memory of the token of $Y$, which is a contradiction with $Pref_s(\cM_X) \prec R$.}

\item{ {Case~4.} In round $k< t'$ the explorer of the token $X$ met an agent $Y$ in state {\tt token} of equal seniority but such that $R \prec Pref_s(\cM_Y)$ where $\cM_Y$ is the memory of agent $Y$, $s$ is the last round before $k$ when the explorer of $X$ updated its variable $recent$-$token$, and $R$ is the variable $recent$-$token$ of the explorer of $X$ in round $k$. Similarly as before, we can get a contradiction with the fact that $R \prec Pref_s(\cM_Y)$.}

\end{itemize}

Since an agent with the largest memory {in $Z$} in a given round must have had the largest memory among the agents in $Z$ in all previous rounds, the claim follows. \finclaim

In order to prove our upper bound on the time spent by $A$ in state {\tt cruiser}, observe that
after at most {$T(SIGN(N))$ rounds} since $A$ transits to state {\tt cruiser}, all other agents have quit state {\tt setup}. 
Consider the additional $2P(N,\ell)$ rounds during which agent $A$ performs $TZ(\ell)$. Let round $\tau$ be the end of the first half
of this segment $S$ of  $2P(N,\ell)$ rounds. Some meeting must have occurred on or before round $\tau$, due to the properties of $TZ$.
If agent $A$ was involved in one of those meetings, it left state {\tt cruiser} by round $\tau$. Otherwise, it must have met some other agent
in state either {\tt cruiser} or {\tt token} during the second half of the segment $S$. Indeed, if it does not meet another agent in state {\tt cruiser},
it must meet another agent in state {\tt token}, which transited to this state by round $\tau$. (Claim 1 guarantees the existence of such an 
agent after round $\tau$.) This proves our upper bound on the time spent by $A$ in state {\tt cruiser}.

From state {\tt cruiser} agent $A$ can transit to one of the three states: {\tt shadow}, {\tt explorer} or {\tt token}. {To deal with the state {\tt shadow}, we need the following claim.}

\vspace*{0.3cm}
\noindent
{{\bf Claim 2.} If agent $A$ becomes the {\tt shadow} of an agent $B$ in some round $t$, then agent $B$ cannot itself switch to state {\tt shadow} in the same round.}

{To prove the claim, there are $3$ cases to consider.}

\begin{itemize}
\item{ {Case~1.} Agent $A$ transits from state {\tt token} to state {\tt shadow} in round $t$. According to the algorithm, agent $B$ is an 
{\tt explorer} transiting to state {\tt searcher} in round $t$.}

\item{{Case~2.} Agent $A$ transits from state {\tt searcher} to state {\tt shadow} in round $t$. According to the algorithm, $B$ is }
\begin{itemize}
\item{ either an agent in state {\tt token} that is not with its explorer in round $t$, in which case agent $B$ remains in state {\tt token} in round $t$}
\item{ or an agent in state {\tt explorer}. However an agent in state {\tt explorer} cannot switch directly to state {\tt shadow}. Hence $B$ cannot transit to
state {\tt shadow} in round $t$. }
\end{itemize}
\item{ {Case~3.} Agent $A$ transits from state {\tt cruiser} to state {\tt shadow} in round $t$. According to the algorithm, $B$ is either in one of the situations described in Case~2, in which case $B$ does not switch to state {\tt shadow} in round $t$, or $B$ is an agent in state {\tt cruiser} that transits to state {\tt explorer}.} 
\end{itemize}

{In all cases, $B$ does not switch to state {\tt shadow} in round $t$, which proves the claim.} \finclaim

{In view of Claim~2 and of the fact that} the termination  conditions for an agent in state 
 {\tt shadow} are the same as of its guide, we may eliminate the case of state  {\tt shadow} from our analysis. 

{Consider an agent $A$ in state {\tt explorer}}. After 
  waiting time of  {$T(SIGN(N))$  $+P(N,L)$ rounds}, where $L$ is the largest possible label (it is polynomial in $N$),
 agent $A$ knows that all other agents have already transited from the state {\tt cruiser}
 (they used at most {$T(SIGN(N))$} rounds in state {\tt setup} and at most  $P(N,L)$ rounds in state  {\tt cruiser}, as their labels are at most $L$ and 
 at least one token is already present in the graph).

{In what follows, we show that, after at most a polynomial time $\rho$, agent $A$ either leaves state {\tt explorer} or declares that gathering is over.}  
 
 %Either agent $A$ never leaves state {\tt explorer}, in which case we will prove that it declares that gathering is over after polynomial time, in some round $\rho$,
 %or it transits to state {\tt searcher} before round $\rho$, 
 %in which case it uses at most $2\cdot T(EXPLO(N))$ rounds for one execution of $EXPLO^*(N)$ and after
 %additional at most  $2\cdot T(EXPLO(N))$ rounds it finds an idle agent in state {\tt token} and becomes
 %its shadow (claim 1 guarantees the existence of such an agent).
 
 {In order to prove this, we first compute an upper bound on the number of non-clean explorations $EXPLO^*(N)$ that can be performed by agent $A$ as an explorer. An exploration could be non-clean due to several reasons,
 according to the description of the algorithm.}
 \begin{itemize}
 \item
{{In round $t'$} agent $A$ met an agent $C$ in state {\tt token} of higher seniority than that of $A$, or of equal seniority but such
  that $recent\mbox{-}token\prec Pref_t( \cM_C)$, {where $t$ is the last round before round $t'$ when the variable $recent$-$token$ of $A$ was updated}. According to the algorithm, agent $A$ transits to state {\tt searcher} as soon as it terminates its exploration $EXPLO^*(N)$ after round $t'$. Hence such a meeting can make at most $1$ non-clean exploration.}
 \item
{In round $t'$ the token $B$ of $A$ was visited  by an agent $C$ in state {\tt explorer} of higher seniority than that of $B$,
     or of equal seniority but such that {$Pref_t(\cM_B) \prec R$}, where $R$ is the variable $recent$-$token$ of agent $C$ and $t$ is the last round before round $t'$ when the variable $recent$-$token$ of $C$ was updated. According to the algorithm, agent $A$ transits to state {\tt searcher} as soon as it terminates its exploration $EXPLO^*(N)$ after round $t'$. Hence such a meeting can make at most $1$ non-clean exploration.}
  \item
  Either agent $A$ or its token $B$ met an agent in state {\tt searcher}. Since the lifespan of a searcher is at most the time of two consecutive executions of 
  $EXPLO^*(N)$, it can overlap at most three consecutive executions of this procedure. Hence one searcher can make non-clean at most 6 explorations 
  (3 by meeting $A$ and 3 by meeting $B$). Since there are at most $N$ searchers, this gives at most $6N$ non-clean explorations. 
   \item
 {In round $t'$} agent $A$ met an agent $C$  in state {\tt token} of lower seniority than that of $A$, or of equal seniority but such that {$ Pref_t(\cM_C) \prec recent\mbox{-}token$, where $t$ is the last round before $t'$ when the variable $recent$-$token$ of $A$ was updated}. After this meeting,
   the remaining time when agent $C$ remains in state {\tt token} is at most the duration of one execution of $EXPLO^*(N)$ (after at most this time the explorer of $C$ becomes searcher and hence $C$ transits to state {\tt shadow}). This time can overlap at most two consecutive executions of  $EXPLO^*(N)$,
   hence such meetings can make at most $2N$ non-clean explorations.  
    \item
    {In round $t'$} the token $B$ of $A$ met an agent $C$  in state {\tt explorer} of lower seniority than that of $B$, or of equal seniority but such that $recent\mbox{-}token \prec Pref_t(\cM_B)$ {(where $t$ is the last round before round $t'$ when the variable $recent$-$token$ of $C$ was updated)}.
     A similar analysis as in the previous case shows that such meetings can make at most $2N$ non-clean explorations. 
     \item
     Agent $A$ met an agent $C$ in state {\tt explorer}. The memories of the two agents at this time are different. After this meeting,
   the remaining time when agent $C$ remains in state {\tt explorer} is at most the duration of two consecutive executions of $EXPLO^*(N)$ because after the 
   return of $C$ on its token, the tokens of $A$ and $C$ have different memories and hence after another exploration, $C$ must become a searcher.
   Indeed, since by assumption $A$ remains in state {\tt explorer} till the end of the algorithm, we must have $R \prec Pref_t(\cM_B)$, where $R$ is the variable $recent$-$token$ of $C$ at the time $t$, where $t$ is the first round after the meeting of $A$ and $C$, in which agent $C$ updated its variable $recent$-$token$. 
  This gives at most $3N$ non-clean explorations. 
   \item
   $A$ met an agent $C$  in state {\tt token} in round $s$, that looked like its token $B$ at this time, but that turned out not to be the token $B$ after the backtrack
    of $A$ on $B$. More precisely, $recent\mbox{-}token=Pref_t(\cM_c)$ in round $s$ (where $t$ is the last round {before round $s$} when the variable $recent$-$token$ of $A$ was updated)
    but $Pref_s(\cM_B) \neq Pref_s(\cM_C)$. After round $s$ agent $C$ remains in state {\tt token} for at most the duration of two executions of $EXPLO^*(N)$.
    This gives at most $3N$ non-clean explorations.
    \item
   {In round $t'$ the token $B$ was visited by an agent $C$ of equal seniority in state {\tt explorer} such that $recent\mbox{-}token = Pref_t(\cM_B)$, where $t$ is the last round before $t'$ when the variable $recent$-$token$ of $C$ was updated and this agent turned out not to be $A$ after 
    the backtrack of $A$ on $B$. Similarly as before, this gives at most $3N$ non-clean explorations.}
     \end{itemize}

      Hence there can be at most {$19N+2$} non-clean executions of $EXPLO^*(N)$ for agent $A$ 
     (notice that, e.g., an agent can make non-clean one exploration in the state 
     {\tt explorer} and then in the state {\tt searcher}, hence for simplicity we add all the above upper bounds). A similar analysis shows that during at most {$19N+2$}
     waiting periods of a duration $2T(EXPLO(N))$ agent $A$ can be met by a new agent. Recall that before performing the first  execution of $EXPLO^*(N)$ agent $A$ has been waiting for {$T(SIGN(N))+P(N,L)$ rounds}.
     {Hence if agent $A$ has not left state {\tt explorer} after at most $\rho=T(SIGN(N))+P(N,L) +(38N+5)(4T(EXPLO(N))$ rounds since it transited to state {\tt explorer}, there
     has been a clean execution of $EXPLO^*(N)$ followed by a waiting period without any new agent coming during the period of $\rho$ rounds, and thus agent $A$ declares that gathering is over by the end of this period. Otherwise, agent $A$ transits to state {\tt searcher} before spending $\rho$ rounds in state 
  {\tt explorer}, in which case it uses at most $2\cdot T(EXPLO(N))$ rounds for one execution of $EXPLO^*(N)$ and after
 additional at most  $2\cdot T(EXPLO(N))$ rounds it finds an idle agent $X$ in state {\tt token} (claim 1 guarantees the existence of such an agent): it then becomes 
 the shadow of either $X$ or of the explorer of $X$.}

  It remains to consider an agent $A$ in state {\tt token}. 
  From this state, either at some point the agent transits to state {\tt shadow} or it remains in state {\tt token} till the end of the algorithm. 
  In this latter case, its explorer declares that gathering is over after at most {$T(SIGN(N))+P(N,L) +(38N+5)(4T(EXPLO(N))$ rounds} since it transited to state 
   {\tt explorer}. However, as soon as an explorer declares that gathering is over, its token does the same. So, agent $A$ declares that gathering is over after at most  
 {$T(SIGN(N))+P(N,L) +(38N+5)(4T(EXPLO(N))$ rounds} since it transited to state {\tt token} (recall that, according to the algorithm, agent $A$ and its explorer have reached their current state at the same time).
   
   Hence every agent eventually terminates.
   We conclude by observing that the execution time of the entire algorithm is upper bounded by the sum of the following upper bounds:
   \begin{itemize}
   \item
   the time between the wake up of the first agent and the time of the wake up of an agent $A$ that will be in state {\tt explorer} when declaring that gathering is over; this time is upper bounded by {$T(SIGN(N))$}.
   \item
   the time that such an agent $A$ spends in state {\tt setup} and {\tt cruiser}
   \item
   the time that such an agent $A$ spends in state {\tt explorer}
   \end{itemize}
   
   We have shown above that each of these upper bounds is {$O(T(SIGN(N))+P(N,L)+N \cdot T(EXPLO(N)))$}, where $L$ is polynomial in $N$. Since the values of
 $T(EXPLO(N))$, $T(SIGN(N))$ and $P(N,L)$ are all polynomial in $N$, this proves that the running time of  Algorithm Gathering-with-Detection is polynomial in $N$.    
\end{proof}

In the sequel we will use the following notion, which is a generalization of the enhanced view of a node. Consider a configuration of agents in any round.
Color nodes $v$ and $w$ with the same color if and only if they are occupied by agents $A_1,\dots ,A_r$ and $B_1, \dots ,B_r$, respectively, where $A_i$ and $B_i$
have the same memory in this round. A {\em colored} view from node $v$ is the view from $v$ in which nodes are colored according to the above rule. 

In view of Lemma \ref{term}, all agents eventually declare that gathering is over. Hence the final configuration must consist of agents in states  {\tt explorer} , 
{\tt token} and {\tt shadow}, all situated in nodes $v_1, \dots ,v_k$, such that in each node $v_i$ there is exactly one agent $E_i$ in state {\tt explorer}, exactly one
agent $T_i$ in state {\tt token}  and possibly some agents in state  {\tt shadow}. Call such a final configuration a {\em clone} configuration if there are at least two
distinct nodes $v_i$, $v_j$ which have identical colored views. We will first show that the final configuration cannot be a clone configuration
and then that it must consist of all agents
gathered in a unique node and hence our algorithm is correct.

 \begin{lemma}\label{clone}
 The final configuration cannot be a clone configuration.
\end{lemma}

\begin{proof}
Suppose for contradiction that the final configuration in round $f$ contains distinct nodes which have identical colored views.
Let $A$ be one of the agents woken up earliest by the adversary. There exists an agent $A'$ (also woken up earliest by the adversary) which has an
identical memory as $A$ and an identical colored view.
Notice that if two agents have the same memory at time $t$ they must have had the same memory at time $t-1$. Since colors in a colored view
are decided by memories of agents, this implies (by a backward induction on the round number) that the colored views of $A$ and $A'$ are the same in each round 
after their wake-up, and in particular {\em in} the round of their wake-up. In this round no agent has moved yet and hence each agent is in a different node. Hence colored
views in this round correspond to enhanced views. Thus we can conclude that the enhanced views from the initial positions of agents $A$ and $A'$ were identical,
which contradicts the assumption that in the initial configuration every agent has a unique enhanced view.
\end{proof}

  \begin{lemma}\label{one}
   In the final configuration all agents must be at the same node.
   \end{lemma}

\begin{proof}
It follows from the formulation of the algorithm that at least one agent transits to state {\tt token}. By Claim 1 in the proof of Lemma \ref{term}, 
there exists an agent that remains in the state {\tt token} till the end of the algorithm. By Lemma \ref{term}, this agent declares that gathering is over.
Let $B$ be the first (or one of the first) agents in state {\tt token} that declares that gathering is over. Let $A$ be its explorer. Let $\tau_0$ be the round in which
agent $A$ starts its last exploration $EXPLO^*(N)$. Let $\tau_{1/2}$ be the round in which backtrack begins during this execution. Let $\tau_1$ be the round in which
this backtrack (and hence the execution of $EXPLO^*(N)$) is finished, and let $\tau_2$ be the round in which $A$ declares that gathering is over.

\vspace*{0.3cm}
\noindent
{\bf Claim 1.} In round $\tau_0$ all agents in state {\tt token} have the same memory.

In order to prove the claim we first show that all agents in state {\tt token} in round $\tau_0$ have the same seniority. Observe that there cannot be any agent
in state {\tt token} of higher seniority than $B$: {at least one of such agents} would be seen by $A$ during its last clean exploration $EXPLO^*(N)$ between rounds $\tau_0$
and $\tau_1$ contradicting its cleanliness. Also there cannot be any agent $C$ in state {\tt token} of lower seniority than $B$. Indeed, let $D$ be the explorer of $C$.
Either $D$ becomes a {\tt searcher} between $\tau_0$ and $\tau_1$ and thus it meets the token $B$ before time $ \tau_2$ which contradicts the
declaration of $A$ and $B$ at time $\tau_2$ or it remains an {\tt explorer}, in which case $C$ remains a {\tt token} between $\tau_0$ and $\tau_1$ and thus $C$ is visited
by $A$ during its last clean exploration,  contradicting its cleanliness. This shows that all agents in state {\tt token} in round $\tau_0$ have the same seniority.
Hence their explorers start and finish $EXPLO^*(N)$ at the same time. Consequently no token existing in round $\tau_0$ can transit to state {\tt shadow}
before round $\tau_1$. Agent $A$ must have seen all these tokens during its last exploration. It follows that the memory of each such token in round $\tau_0$ must
be equal to the memory of $B$ at this time: otherwise, agent $A$ would detect such a discrepancy during its last exploration, which would contradict the 
cleanliness of this exploration. This proves Claim 1.\finclaim

Claim 1 implies that in time $\tau_0$ all agents in state {\tt explorer} have the same memory. Indeed, since at time $\tau_0$ agent $A$ is together with $B$, each explorer
must be with its token, since tokens have the same memory.

 \vspace*{0.3cm}
  \noindent
{\bf Claim 2.} In round $\tau_0$ there are no agents in state {\tt searcher}.

Suppose for contradiction that there is a searcher $S$ in round $\tau_0$. Recall that $S$ performs two explorations: one entire exploration $EXPLO^*(N)$ 
and another partial exploration  $EXPLO^*(N)$ until meeting a token or an explorer.

  \noindent
{\em Case 1.} $S$ finished its first exploration $EXPLO^*(N)$ by round $\tau_0$.\\
Hence its second exploration ends by round $\tau_1$. It could not end by round $\tau_0$ because $S$ would not be a searcher in this round anymore.
If it ended between $\tau_0$ and $\tau_1$, it must have met a token $C$. By Claim 1, all explorers have the same seniority and hence at time $\tau_1$ the 
explorer $D$ of $C$ backtracked to $C$. This exploration is not clean for $D$. Either $D$ becomes a {\tt searcher} at time $\tau _1$ and thus meets $A$ and $B$
before time $\tau_2$, contradicting their declaration at time $\tau_2$, or $D$ starts another  $EXPLO^*(N)$ and it meets itself $A$ and $B$
before time $\tau_2$, contradicting their declaration at time $\tau_2$. This shows that the second exploration of $S$ cannot end between $\tau_0$ and $\tau_1$,
hence Case 1 is impossible.

  \noindent
{\em Case 2.} $S$ finished its first exploration $EXPLO^*(N)$ between $\tau_0$ and $\tau_{1/2}$.\\
Hence it must visit some token $C$ during its second exploration (and before starting the backtrack) by round $\tau_1$. As before, this contradicts 
the declaration of $A$ and $B$ at time $\tau_2$.

  \noindent
{\em Case 3.} $S$ finished its first exploration $EXPLO^*(N)$ between  $\tau_{1/2}$ and $\tau_1$.\\
Hence the entire backtrack during this first exploration took place between rounds  $\tau_0$ and $\tau_1$. During this backtrack, $S$ visited some token.
As before, this contradicts 
the declaration of $A$ and $B$ at time $\tau_2$.

  \noindent
{\em Case 4.} $S$ finished its first exploration $EXPLO^*(N)$ after round $\tau_1$.\\
This is impossible, as it would not be in state {\tt searcher} in round $\tau_0$.

This concludes the proof of Claim 2.\finclaim

\vspace*{0.3cm}
\noindent
{\bf Claim 3.} Let $\E$ be the set of agents in state {\tt explorer} in round $\tau_0$. 
In round $\tau_{1/2}$ every agent from $\E$ can reconstruct its colored view in round  $\tau_0$.

To prove the claim first note that since agent $A$ starts its last exploration in round $\tau_0$ and all agents from $\E$ have the same memory in round $\tau_0$,
they all start an exploration $EXPLO^*(N)$ in this round. In round $\tau_{1/2}$ every agent from $\E$ has visited all nodes of the graph and starts its backtrack.
In round $\tau_0$ there are no agents in state {\tt setup} or {\tt cruiser}, in view of the waiting time when $A$ transited to state {\tt explorer}, and there are no agents
in state {\tt searcher} by Claim 2. Hence the visit of all nodes between rounds $\tau_0$ and $\tau_{1/2}$ permits to see all agents that were tokens at time $\tau_0$.
Since at this time every explorer were with its token, this
permits to reconstruct the memories and the positions of all agents in round $\tau_0$. This is enough to reconstruct the colored views of all agents in round  $\tau_0$, which proves the claim.\finclaim

To conclude the proof of the lemma it is enough to show that in round $\tau_0$ only one node is occupied by agents, since this will be the final configuration.
Suppose that nodes $v\neq v'$ are occupied in this round. Let $A$ be the explorer at $v$ and $A'$ the explorer at $v'$. Note that the colored views of $A$ and $A'$
in round $\tau_0$
must be different, for otherwise the configuration in round $\tau_0$ would be a clone configuration, and consequently the final configuration would also be clone,
contradicting Lemma \ref{clone}. Since, by Claim 3, in round $\tau_{1/2}$ each of the agents $A$ and $A'$ has reconstructed its colored view in round  $\tau_0$, their memories in round $\tau_{1/2}$ are different. Between rounds  $\tau_{1/2}$ and $\tau_1$, during its backtrack, agent $A'$ has visited again all tokens, 
in particular the token of $A$. Hence $A$, after backtracking
to its token in round $\tau_1$,  realizes that another explorer has visited its token, which contradicts the cleanliness of the last exploration of $A$.
This contradiction shows that in round $\tau_0$ only one node is occupied and hence the same is true in the final configuration. This concludes the proof of the lemma.
\end{proof}

Now the proof of Theorem \ref{eq} follows directly from Lemmas  \ref{not}, \ref{term},  and \ref{one}.

\section{Unknown upper bound on the size of the graph}

In this section we show that, if no upper bound on the size of the graph is known, then there is no universal algorithm for gathering
{\em with detection} all gatherable configurations. Nevertheless, we still show in this case a universal algorithm that gathers all gatherable configurations:
all agents from any gatherable configuration eventually stop forever at the same node (although no agent is ever sure
that gathering is over). The time of such an algorithm is the number of rounds between the wake-up of the first agent and the last round in which some agent moves.
Our algorithm is polynomial in the (unknown) size of the graph.

We first prove the following negative result.

\begin{theorem}\label{no}
There is no universal algorithm for gathering with detection all gatherable configurations in all graphs. 
\end{theorem}

\begin{proof}
Consider the following initial configurations. In configuration $C$ the graph is a 4-cycle with clockwise oriented ports 0,1 at each node, and with additional
nodes of degree 1 attached to two non-consecutive nodes. There are two agents starting at a node of degree 2 and at its clockwise neighbor, cf. Fig. 
 \ref{fig:gatherable} (a).
In configuration $D_n$, for $n=4k$, the graph is constructed as follows. Take a cycle of size $n$  with clockwise oriented ports 0,1 at each node. 
Call clockwise consecutive nodes of the cycle $v_0,\dots, v_{n-1}$ (names are used only to explain the construction) 
and attach two nodes of degree 1 to $v_0$ and one node of degree 1 to every other node with even index. Initial positions of agents are at nodes $v_i$, where 
$i=4j$ or $i=4j-1$, for some $j$,  cf. Fig. \ref{fig:gatherable} (b).

\begin{figure}[h!]
        \begin{center}
        \includegraphics[width=0.8\textwidth]{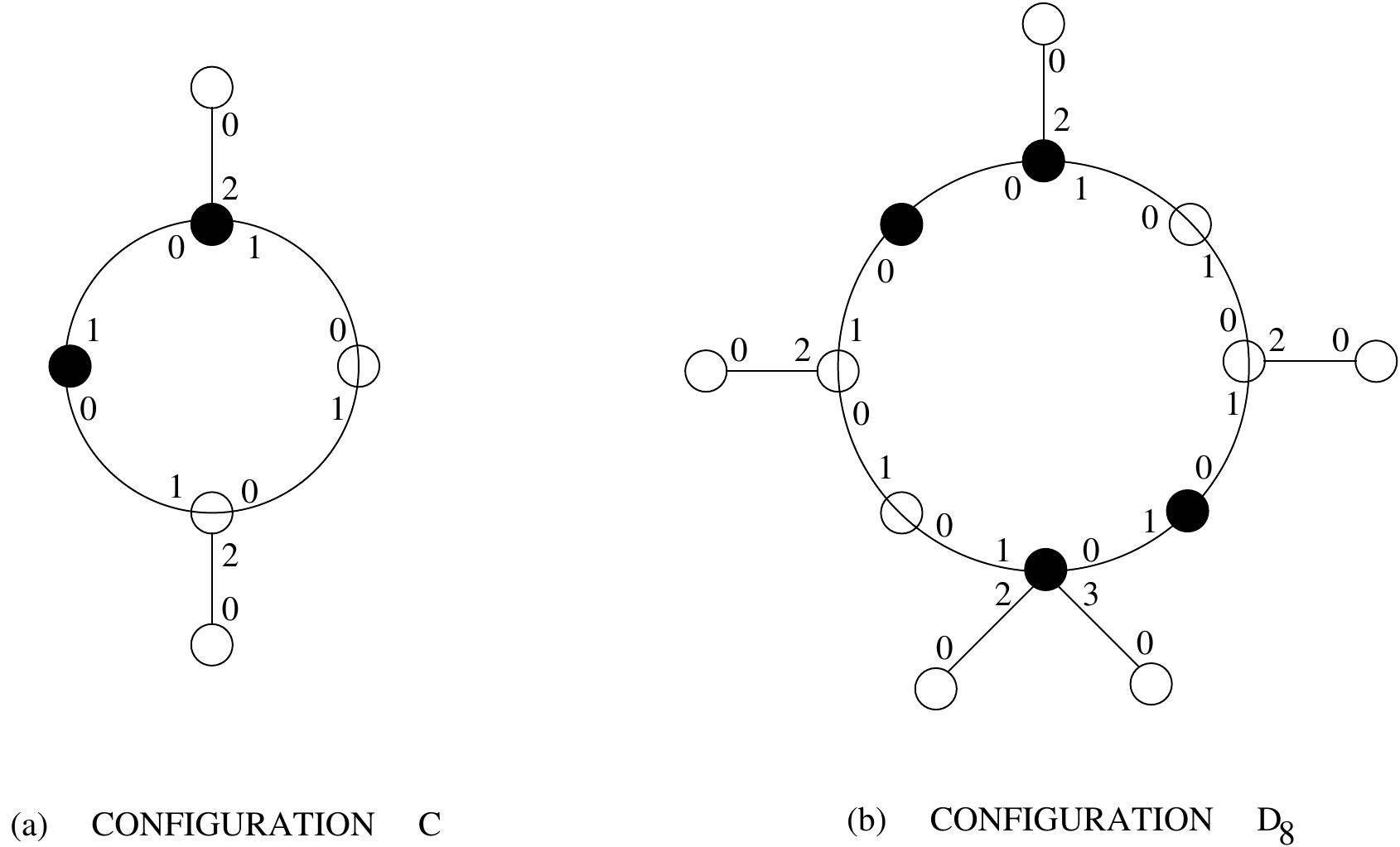}
        \caption{Configurations $C$ and $D_8$ in the proof of
Theorem \ref{no}. Black nodes are occupied by agents}
        \label{fig:gatherable}
        \end{center}
\end{figure}

Each of the configurations $C$ and $D_n$, for $n \geq 8$, is gatherable. Indeed, in each of these configurations there exist agents with different views
(agents starting at nodes of degree 2 and of degree 3) and each agent has a unique enhanced view (this is obvious for configuration $C$ and follows 
from the existence of a unique node of degree 4 for configurations $D_n$). Hence each of these configurations satisfies condition {\bf G} and consequently,
by Theorem \ref{eq}, there is an algorithm for gathering with detection each specific configuration, as such a dedicated algorithm knows 
the configuration and hence may use the knowledge of the size
of the graph.

It remains to show that there is no {\em universal} algorithm that gathers with detection all configurations $C$ and $D_n$. Suppose, for contradiction, that $\cA$ is such 
an algorithm. Suppose that the adversary wakes up all agents simultaneously and let $t$ be the time after which agents in configuration $C$ stop at the same node
and declare that gathering is over. Consider the configuration $D_{8t}$ and two consecutive agents antipodal to the unique node of degree 4, i.e.,   starting from nodes
$v_{4t}$ and $v_{4t-1}$. Call $X$ the agent starting at a node of degree 2 in configuration $C$ and call $Y$ the agent starting at its clockwise neighbor (of degree 3)
in this configuration. Call $X'$ the agent starting at node $v_{4t-1}$ and call $Y'$ the agent starting at node $v_{4t}$ in configuration $D_{8t}$.
(Again names are used only to explain the construction.) 

In the first $t$ rounds of the executions of algorithm $\cA$ starting from configurations $C$ and $D_{8t}$ the memories of the agents $X$ and $X'$ and
of the agents $Y$ and $Y'$ are 
the same. This easily follows by induction on the round number. Hence after $t$ rounds agents $X'$ and $Y'$ starting from configuration  $D_{8t}$ stop and (falsely)
declare that gathering is over. This contradicts universality of algorithm $\cA$.
\end{proof}

Our final result is a universal algorithm gathering all gatherable configurations, working without any additional knowledge. It accomplishes correct gathering
and always terminates but (as opposed to Algorithm Gathering-with-Detection which used an upper bound on the size of the graph), this algorithm does not
have the feature of detecting that gathering is over. We first present a high-level idea of the algorithm, then describe it in detail and prove its correctness.
Recall that we assume that the initial configuration satisfies condition {\bf G} (otherwise gathering, even without detection,  is impossible by Lemma \ref{not}).

 \vspace*{0.2cm}

\noindent
{\bf Idea of the algorithm.}

Since in our present scenario no upper bound on the size of the graph is known, already guaranteeing any meeting between agents must be done 
differently than in Algorithm Gathering-with-Detection. After wake-up each agent proceeds in phases {$i=1,2,\dots$}, where in phase $i$ it ``supposes'' that the graph
has size at most $2^i$.  In each phase an appropriate label based on procedure $SIGN(2^i)$ is computed and procedure $TZ$ is performed sufficiently long to guarantee
a meeting at most at the end of phase {$\lceil \log_2 m \rceil$}, where $m$ is the real size of the graph. If no meeting occurs in some phase for a sufficiently long time, the agent starts the next phase.

Another important difference occurs after the meeting, when one of the agents becomes an explorer and the other its token.  Unlike in the case of known upper bound
on the size of the graph, 
there is no way for any explorer to be sure at any point of the execution that it has already visited the entire graph. Clearly procedure $EXPLO(m)$ cannot give this guarantee, as $m$ is unknown, and procedure $EST$ of exploration with a stationary token, which does not require the knowledge of an upper bound, cannot give this guarantee either, as an explorer cannot be always sure that it visits its own token, because memories of several agents playing the role of the token can be 
identical at various stages of the execution, and hence these ``tokens'' may be undistinguishable for the explorer. 

Nevertheless, our algorithm succeeds in accomplishing the task by using a mechanism which is analogous to the ``butterfly effect''.
Even a slight asymmetry in a remote part of the graph is eventually communicated to all agents and
guarantees that at some point some explorer will visit the entire graph (although in some graphs no explorer can ever be sure of it at any point of an execution) and then all agents will eventually gather at the token of one of these explorers. Making all agents decide on the same token uses property {\bf G}
and is one of the main technical difficulties of the algorithm.

 \vspace*{0.2cm}
 
  \noindent
{\bf Algorithm Gathering-without-Detection}

 Similarly as in Algorithm Gathering-with-Detection,  an agent can be in one of the following five states: {\tt traveler}, {\tt shadow}, {\tt explorer},  {\tt token}, {\tt searcher}.
 State {\tt traveler} partly combines the roles of previous states {\tt setup} and {\tt cruiser}. 
 For every agent $A$ in state  {\tt shadow} the notion of guide is defined as before. 
 Below we describe the actions of an agent $A$ in each of the states and the transitions between the states. At wake-up agent $A$ enters the state  {\tt traveler}.
 
  \vspace*{0.2cm}
 
  \noindent
  {\bf State} {\tt traveler}.
  
  {In this state agent $A$ works in phases  {$i=1,2,\dots$}. In phase $i$ the agent supposes that the graph has size at most $n=2^i$.
 {Agent $A$ performs $SIGN(2^i)$ in order to visit all nodes (and wake up all still dormant agents), if the assumption is correct, and find the current signature of its initial position $v$, called the {\em label} $\ell_{n}=\ell_{2^i}$ of agent $A$.}
 Let $L_{2^i}$ be the maximum possible label of an agent in phase $i$. (Note that $L_{2^i}$ is polynomial in $n$).
 Then  agent $A$ performs $TZ(\ell_n)$ for $\Delta_n$ rounds, where {$\Delta_n=\Delta_{2^i}=T(SIGN(2^i))+2P(2^i,L_{2^i})+ \Sigma_{j=1}^{i-1}Q_j $}, for {$n\geq 4$ (i.e., $i\geq2$), and $\Delta_2=T(SIGN(2))+2P(2,L_2)$. In the formula for $\Delta_n$, $Q_j$ is defined as {$T(SIGN(2^j))+\Delta_{2^j}$} and is an upper bound on the duration of phase $j$. Note that by induction on $i$ we can prove that $\Delta_{2^i}= T(SIGN(2^i))+2P(2^i,L_{2^i}) + \Sigma_{j=1}^{i-1}[2^j(T(SIGN(2^{(i-j)}))+P(2^{(i-j)},L_{2^{(i-j)}}))]$. Hence $\Delta_{n}$ is upper-bounded by\\ $2^ii(T(SIGN(2^i))+2P(2^i,L_{2^i}))=n\log_2(n)(T(SIGN(n))+2P(n,L_{n}))$ which is a polynomial in $n$.}}

 {If no agent has been met during phase $i$, agent $A$ starts phase $i+1$.} As soon as another agent is met in some phase $k$, agent $A$
 interrupts this phase and transits either to state
 {\tt shadow} or {\tt token} or {\tt explorer}. Suppose that the first meeting of agent $A$ occurs in round $t$ at node $v$.
 
 \noindent
 Case 1. There are some agents  in round $t$ at node $v$ which are either in state  {\tt searcher}, or {\tt explorer} or  {\tt token} .\\
 Let $\cH$ be the set of these agents.
 
{ Subcase 1.1.
 There are some agents in $\cH$ that are either in state {\tt explorer} or {\tt token}. Let $\mathcal{I}$ be the set of all those agents in $\cH$.
 Agent $A$ transits to state {\tt shadow} and its guide is the agent having the largest memory in set $\mathcal{I}$.}
 
{Subcase 1.2.
 There is no agent in $\cH$ that is either in state {\tt explorer} or {\tt token}. Agent $A$ transits to state {\tt shadow} and its guide is the agent in state {\tt searcher} having the largest memory in set $\cH$.}

 \noindent
 Case 2. There are only agents in state {\tt traveler} in round $t$ at node $v$.
 
 Subcase 2.1.
 Agent $A$  has the largest memory among all agents in round $t$ at node $v$.\\
 Then agent $A$  transits to state {\tt explorer}.
 
 Subcase 2.2.
 Agent $A$  does not have the largest memory among all agents in round $t$ at node $v$.\\
 If there is exactly one agent $B$  with memory larger than $A$, then agent $A$ transits to state {\tt token}.
 Otherwise, it transits to state {\tt shadow} of the agent  with largest memory.
 
 (Note that cases 1 and 2 cover all possibilities because 
 an agent in state {\tt shadow} always accompanies its guide and this guide cannot be an agent in state {\tt traveler}.)
 
 \vspace*{0.2cm}
  \noindent
 {\bf State} {\tt shadow}.
 
 Agent $A$ has exactly one guide and is at  the same node as the guide in every round. In every round it makes the same move as the guide.
 If the guide $B$ transits itself to state  {\tt shadow} and gets agent $C$ as its guide, then agent $A$ changes its guide to $C$ as well.  

In the description of the actions in the three remaining states, we will use the notion of seniority defined for  Algorithm Gathering-with-Detection.

    \vspace*{0.2cm}
 
  \noindent
    {\bf State} {\tt explorer}.
    
    When agent $A$ transits to state {\tt explorer}, there is another agent $B$ that transits to state  {\tt token} in the same round at the same node $v$.
 Agent $B$ is called the token of $A$. Agent $A$ has a variable $recent$-$token$ that it initializes to the memory of $B$ in this round. 
 
 We first define the notion of a {\em consistent meeting} for agent $A$.
 Let $t$ be the last round when agent $A$ updated its variable $recent$-$token$. 
 A consistent meeting for $A$ is a meeting in round $t'>t$ with an agent $C$ in state  {\tt token}
 of the same seniority as $A$,  such that $\cM$ is the current memory of $C$ and  $Pref_t(\cM)=recent\mbox{-}token$. Intuitively, a consistent meeting is a meeting of
 an agent that $A$ can plausibly consider to be its token $B$. Note that, according to this definition, a meeting in the round when the variable $recent$-$token$ is updated,
 is not a consistent meeting.
 
 We now briefly describe the procedure $EST$ based on \cite{CDK} that will be subsequently adapted to our needs {and which allows an agent to construct a BFS tree of the network provided that it cannot confuse its token with another one.}
The agent constructs a BFS tree rooted at its starting node $r$
marked by the stationary token. In this tree it marks port numbers at all nodes. 
During the BFS traversal, some nodes are added to the BFS tree. {In the beginning, the agent adds the root $r$ and then it makes the {\em process} of $r$. The process of a node $w$ consists in checking all the neighbors of $w$ in order to determine whether 
some of them have to be added to the tree or not. When an agent starts the process of a node $w$, it goes to the neighbor reachable via port $0$ and then checks the neighbor.}

{When a neighbor $v$ of $w$ gets checked, the agent  
verifies if $v$ is equal to some node
previously added to the tree. To do this, for each node $u$ belonging to the current BFS tree, the agent travels from $v$ using the reversal $\overline{q}$ of the shortest path $q$ from $r$ to $u$ in the BFS tree (the path $q$ is
a sequence of port numbers). If at the end of this backtrack it meets the token, then $v=u$: in this case $v$ is not added to the tree as a neighbor of $w$ and is called $w$-{\em rejected}. If not, then $v\neq u$. Whether node $v$ is rejected or not, the agent then comes back to $v$ using the path $q$. If $v$ is different from all the nodes of the BFS tree, then it is added to the tree.}

{Once node $v$ is added to the tree or rejected, the agent makes an edge traversal in order to be located at $w$ and then goes to a non-checked neighbor of $w$, if any. The order, in which the neighbors of $w$ are checked, follows the increasing order of the port numbers of $w$.} 

{When all the neighbors of $w$ are checked, the agent proceeds as follows. Let $X$ be the set of the shortest paths in the BFS tree leading from the root $r$ to a node $y$ having non-checked neighbors. If $X$ is empty then procedure $EST$ is completed. Otherwise, the agent goes to the root $r$, using the shortest path from $w$ to $r$ in the BFS tree, and then goes to a node $y$ having non-checked neighbors, using the lexicographically smallest path from $X$. From there, the agent starts the process of $y$.}

{Note that given a graph $\mathcal{G}$ of size at most $n$, every execution of $EST$ in $\mathcal{G}$ lasts at most $8n^5$ rounds. Indeed,  processing every node $v$ takes at most $4n^3$ rounds (because each node $v$ has at most $n-1$ neighbors, checking a neighbor of $v$ takes at most $2n^2$ rounds, and before (resp. after) each checking of a neighbor $w$ of $v$, the agent makes an edge traversal from $v$ to $w$ (resp. $w$ to $v$)). Considering the fact that there are at most $n$ nodes to process in $\mathcal{G}$ and the fact that moving from a node that has been processed to the next node to process costs at most $2n$ rounds, we get the upper bound of $8n^5$ rounds. Hereafter we define $T(EST(n))$ as being equal to $8n^5$.}

{The procedure $EST'$ is a simulation of $EST$ with the following two changes. The first change concerns the beginning of the execution of $EST'$ when the agent is with its token: it updates its variable $recent$-$token$ w.r.t to the current memory of its token. The second change concerns meetings with the token. Consider a verification if a node $w$ , which is getting checked, is equal to some previously constructed node $u$. This verification consists in traveling from $w$ using the reverse path $\overline{q}$,  where $q$ is the path from the root $r$ to $u$ in the BFS tree and checking the presence of the token. If 
 at the end of the simulation of path $\overline{q}$ in $EST'$ agent $A$ makes a consistent meeting, then it acts as if it saw the token in $EST$; otherwise it acts as if it did not see
 the token in $EST$.}

{To introduce the next proposition, we first need to define the notion of a {\em truncated spanning tree}. We say that a tree $T$ is a truncated spanning tree of a graph $\mathcal{G}$ if $T$ can be obtained from a spanning tree of $\mathcal{G}$ by removing one or more of its subtrees.}

{
\begin{proposition}
\label{prop1}
Given a graph $\mathcal{G}$ of size {of at most $n$} (unknown to the agents), the following two properties hold: (1) every execution of $EST'$ in $\mathcal{G}$ lasts at most $T(EST(n))$ rounds and (2) every execution of $EST'$ produces a spanning tree of $\mathcal{G}$ or a truncated spanning tree of $\mathcal{G}$.
\end{proposition}}

\begin{proof}
{If the agent never confuses its token with another one, the proposition follows directly. So in this proof, we focus only on the situations where there are possible confusions among tokens. When such confusions may occur? }

{The execution of procedure $EST'$ consists of alternating periods of two different types. The first one corresponds to periods when the agent processes a node and the second one corresponds to those when the agent moves to the next node to process it. During the periods of the second type, an agent does not use any token to move: it follows the same path regardless of whether it meets some token or not on its path. Hence, an agent can confuse its token with another one only in periods of the first type. During such periods, an agent may indeed be "mislead" by a token which is not its own token, when verifying whether a node has to be rejected or not, by wrongly rejecting a node. This leads to the construction of a spanning tree which is truncated, which proves the second property of the proposition.}

{Concerning the first property, note that, as for procedure $EST$, every execution of $EST'$ in $\mathcal{G}$ lasts at most $8n^5$. Indeed,  processing every node $v$ also takes at most $4n^3$ rounds (as each node $v$ has at most $n-1$ neighbors, checking a neighbor of $v$ takes at most $2n^2$ rounds, and before (resp. after) each checking of a neighbor $w$ of $v$, the agent makes an edge traversal from $v$ to $w$ (resp. $w$ to $v$)). Besides, still for the same reasons as for procedure $EST$, there are at most $n$ nodes to process in $\mathcal{G}$ and moving from a node that has been processed to the next node to process costs at most $2n$ rounds. Hence, we also obtain the upper bound of $8n^5$ rounds. Since $T(EST(n))=8n^5$, the first property of the proposition follows.}
\end{proof}
 
{The procedure $EST^*$ is a simulation of $EST'$ with the following change.}
 Suppose that the execution of {$EST'$} produced the route $\alpha$ of the agent. In procedure $EST^*$, upon completing
 procedure $EST'$, the agent traverses the reverse route $\overline{ \alpha}$ and then again $\alpha$ and $\overline{ \alpha}$. Hence in procedure $EST^*$
 the agent traverses the concatenation of routes $\alpha,\overline{ \alpha}, \alpha,\overline{ \alpha}$. These parts of the trajectory will be called, respectively, the
 first, second, third and fourth segment of $EST^*$.
 The variable $recent$-$token$ is updated at the beginning of the first and third segment of $EST^*$. Note that in these rounds agent $A$ is certain to be with its token. {The (possibly truncated) spanning tree resulting from the simulation $EST^*$ is the (possibly truncated) spanning tree resulting from the execution of the first segment.}

{In view of the description of procedure $EST^*$ and Proposition~\ref{prop1}, we have the following proposition.}

{
\begin{proposition}
\label{prop2}
Given a graph $\mathcal{G}$ of size $n$ (unknown to the agents), the following two properties hold: (1) every execution of $EST^*$ in $\mathcal{G}$ lasts at most $4T(EST(n))$ rounds and (2) every execution of $EST^*$ produces a spanning tree of $\mathcal{G}$ or a truncated spanning tree of $\mathcal{G}$.
\end{proposition}}
   
% The second change
% concerns meetings with the token. Consider a verification if a newly reached node $w$ in $EST$ is equal to some previously constructed node $u$. 
% This verification consists in traveling from $w$ using the reverse path $\overline{q}$,  where $q$ is the path from the root $r$ to $u$ in the BFS tree
% and checking the presence of the token. If 
% at the end of the simulation of path $\overline{q}$ in the first segment of $EST^*$  agent $A$ makes a consistent meeting, then it acts as if it saw the token in %$EST$; otherwise it acts as if it did not see
% the token in $EST$. 
 
Similarly as for $EXPLO^*(n)$,  an execution of $EST^*$  is called {\em clean} if 
the following condition is satisfied: in each round  during this execution, in which $A$ met an agent $C$ {that is not in state {\tt shadow}}, the memory of $C$ is equal to that of $B$,
and in each round during this execution, in which the token $B$ was met by an agent $D$, the memory of $D$ was equal to that of $A$. Notice that 
after the execution of  $EST^*$, agent $A$ is together with its token $B$ and thus they can verify if the execution was clean, by inspecting their memories.
%The execution time of $EST^*$ in a graph of size $m$ (unknown to the agents) is at most $4T(EST(m))$.

 After transiting to state {\tt explorer}, agent $A$ executes the following protocol:
 
%\pagebreak
% \newpage
 \noindent
 {\bf repeat forever}\\ 
 /*Before the first turn of the loop agent $A$ has just entered state {\tt explorer} and is with its token.
 After each turn of the loop, agent $A$ is with its token, waiting  after a clean exploration.*/ 
 
  \vspace*{0.2cm} 
  \noindent
  \hspace*{0.5cm} {\bf if} $A$ has just transited to state {\tt explorer} {\bf or} $A$ has just been visited by another agent {\bf then}\\
   \noindent
  \hspace*{1cm}{\bf do}\\
   \noindent
 \hspace*{1.5cm}$EST^*$
 
  \vspace*{0.2cm}
  \noindent
  \hspace*{1.5cm}{\bf if} {in round $t'$} agent $A$ met an agent $C$ in state {\tt token} of higher\\ 
   \noindent
  \hspace*{1.5cm}seniority than that of $A$ or of equal seniority but 
   such that\\
   \noindent
   \hspace*{1.5cm}$recent\mbox{-}token\prec Pref_t( \cM_C)$ where $\cM_C$ is the memory of agent $C$\\
    \noindent
    \hspace*{1.5cm}and $t$ is the last
   round before round $t'$ when agent $A$ updated its \\
    \noindent
    \hspace*{1.5cm}variable $recent$-$token$ {\bf then} $A$ transits to state {\tt searcher}
    
     \vspace*{0.2cm}
     \noindent
     \hspace*{1.5cm}{\bf if} {in round $t'$ agent} $A$ met another agent $C$ in state {\tt explorer}, such\\ 
      \noindent
    \hspace*{1.5cm}{that either the seniority of $C$ is higher
     than that of $A$, or these}\\ 
      \noindent
      \hspace*{1.5cm}{seniorities are equal but $Pref_{min(t_A,t_C)}(R_A) \prec Pref_{min(t_A,t_C)}(R_C)$,}\\ 
       \noindent
      \hspace*{1.5cm}{where $R_A$ (resp. $R_C$) is the value of the variable $recent$-$token$ of $A$}\\ 
      \noindent
      \hspace*{1.5cm}{(resp. $C$) at the time of the meeting and $t_A$ (resp. $t_C$) is the}\\
 \noindent
      \hspace*{1.5cm}{last round before $t'$ when agent $A$ (resp. $C$)}\\
 \noindent
      \hspace*{1.5cm}{updated its variable $recent$-$token$
       {\bf then} $A$ transits to state {\tt searcher}}
    
     \vspace*{0.2cm}
     \noindent
     \hspace*{1.5cm}{\bf if} $B$ was visited in round $t'$ by an agent $C$ in state {\tt explorer} of\\ 
     \noindent
      \hspace*{1.5cm}higher seniority than that of $B$
     or of equal seniority but such that\\ 
     \noindent
      \hspace*{1.5cm}{$Pref_t(\cM_B) \prec R$, where $\cM_B$ is the memory of agent $B$,
       $R$ is}\\
      \noindent
      \hspace*{1.5cm}{the variable $recent$-$token$ of agent $C$ in round $t'$, and $t$ is the last}\\
      \noindent
      \hspace*{1.5cm}{round before $t'$ when the variable $R$ was updated}\\
    \noindent
      \hspace*{1.5cm}{\bf then} $A$ transits to state {\tt searcher}\\

     \noindent
     \hspace*{1cm}{\bf until} the execution of $EST^*$ is clean

 \vspace*{0.2cm}
        \noindent
     {\bf State}  {\tt token}.
     
      When agent $A$ transits to state {\tt token}, there is another agent $B$ that transits to state  {\tt explorer} in the same round at the same node $v$.
 Agent $B$ is called the explorer of $A$. 
 %Agent $A$ has a variable $recent-explorer$ that it initializes to the memory of $B$ in this round. 
 Agent $A$ remains  idle at a node $v$ and does not change its state, except when its explorer $B$ transits to state {\tt searcher}. In this case it transits
 to state {\tt shadow} and $B$ becomes its guide. 
  
  \vspace*{0.2cm}
        \noindent
    {\bf State} {\tt searcher}
    
After transiting to state {\tt searcher} agent $A$ performs the sequence of explorations $EXPLO(n)$ for {$n=1,2,3,\dots$}, until it meets an agent in state
{\tt token} or {\tt explorer} in round $t$. Let $\cS$ be the set of these agents met by $A$ in round $t$. Agent $A$ transits to state {\tt shadow} and its guide is the 
agent from $\cS$ with largest memory.   

The analysis of the algorithm is split into the following lemmas.

\begin{lemma}\label{term2}
In Algorithm Gathering-without-Detection every agent eventually stops after time polynomial in the size of the graph.
\end{lemma}

\begin{proof}
Let $m$ be the size of the graph (unknown to the agents). Let {$i=\lceil \log_2 m \rceil$}. Let $A$ be any agent. 
We may assume that at some point $A$ is woken up (otherwise it would be idle all the time).
{We will first show that $A$ must meet some other agent at the end of phase $i$ at the latest. To this end, we need to prove the following claim.}

\noindent
{{\bf Claim 1.} Let $t$ be the first round, if any, in which an agent transits to state {\tt token}. Then there exists an agent $B$ that remains in state {\tt token} and is idle from round $t$ on. }

{To prove the claim, let $Z$ be the set of agents that transited to state {\tt token} in round $t$. In every round $t'>t$, the agent from $Z$ with the current largest memory remains in state {\tt token} and stays idle. Indeed, the reasons why such an agent, call it $X$, could leave the state {\tt token} in round $t'$ all lead to a contradiction. There are six such reasons: four of them are identical to those given in Claim~1 of the proof of Lemma~\ref{term}. Hence to show the validity of the claim, we only need to deal with the two remaining reasons which are the following ones.}
\begin{itemize}

\item{ The explorer of token $X$, denoted $E$, met an agent $Y$ in state explorer of higher seniority. Since agent $E$ has the same seniority as agent $X$, by transitivity the seniority of agent $Y$ (and of its token) is higher than that of agent $X$, which is a contradiction with the definition of $Z$.}

\item{ In round $k<t'$ the explorer of token $X$, denoted $E$, met an agent $Y$ in state explorer of equal seniority but such that $Pref_{min(t_E,t_Y)}(R_E) \prec Pref_{min(t_E,t_Y)}(R_Y)$, where $R_E$ (resp. $R_Y$) is the value of the variable $recent$-$token$ of $E$ (resp. $Y$) at the time of the meeting and $t_E$ (resp. $t_Y$) is the last round before $k$ when agent $E$ (resp. $Y$) updated its variable $recent$-$token$. This case is impossible. Indeed, by definition, agent $X$ is among the agents having the highest memory among the agents from $Z$ in round $t'$. Hence, according to the definition of order $\prec$ given in Section~\ref{prelim}, agent $X$ was among the agents having the highest memory among the agents from $Z$ in all rounds between $t$ and $t'$. Moreover, since $E$ and $Y$ have the same seniority, this implies that the token of $Y$ belongs to $Z$. Hence, in round $min(t_E,t_Y)$ the memory of agent $X$ is greater than or equal to the memory of the token of $Y$, which is a contradiction with $Pref_{min(t_E,t_Y)}(R_E) \prec Pref_{min(t_E,t_Y)}(R_Y)$.}

\end{itemize}

{Since an agent with the largest memory in $Z$ in a given round must have had the largest memory among the agents in $Z$ in all previous rounds, the claim follows.} \finclaim

Now we are ready to prove the following claim.

\vspace*{0.3cm}
\noindent
{{\bf Claim 2.} Agent $A$ must meet some other agent at the end of phase $i$ at the latest.}

{Assume by contradiction that agent $A$ does not meet any agent by the end of phase $i$. So, there exists at least one agent executing the first {$i$ phases} in state {\tt traveler}.
Let $F$ be the first agent to finish the execution of phase $i$. According to the algorithm, phase $i$ is made up of two parts.
The first one consists in performing $SIGN(2^i)$ and finding the current signature $\ell_{2^i}$ of the initial position of the executing agent. The signature $\ell_{2^i}$ plays the role of the agent's label in the second part of phase $i$ which consists in performing $TZ(\ell_{2^i})$ for {$\Delta_{2^i}= T(SIGN(2^i))+2P(2^i,L_{2^i})+ \Sigma_{j=1}^{i-1}Q_j$ rounds, where $\Sigma_{j=1}^{i-1}Q_j$ is an upper bound on the sum of durations of phases $1$ to $i-1$, and $L_{2^i}$ is the maximum possible label of an agent in phase $i$}. Observe that at the end of the execution by agent $F$, at some round $t$, of the first part of phase $i$, all the agents in the graph are necessarily woken up due to the properties of procedure $SIGN(2^i)$ (as $m\leq 2^i$). Hence, we consider two cases.}
\begin{itemize}
\item{No agent meets another agent by round {$t+ T(SIGN(2^i)) +\Sigma_{j=1}^{i-1}Q_j$}. In that case, we know that from round {$t+1+ T(SIGN(2^i)) +\Sigma_{j=1}^{i-1}Q_j$} on, all the agents execute the second part of phase $i$ and for each of them there remain at least $P(2^i,L_{2^i})$ rounds before the end of the second part of phase $i$. Since the agents cannot all determine the same signature in phase $i$, there are at least two agents having two distinct labels and thus two agents meet by round {$t'=t + T(SIGN(2^i)) + P(2^i,L_{2^i})+\Sigma_{j=1}^{i-1}Q_j$} due to the properties of procedure $TZ$. So, at least one agent transits to state {\tt token} by round $t'$. However, the last period of $P(2^i,L_{2^i})$ rounds when agent $A$ executes the second part of phase $i$ starts after round $t'$. So, in view of the properties of procedure $TZ$ and of Claim~1, we know that agent $A$ meets an agent in state $token$ by the end of phase $i$ if it does not meet an agent in another state before. We get a contradiction with the assumption made at the beginning of this proof.}

\item{At least two agents meet by round {$t+ T(SIGN(2^i)) + \Sigma_{j=1}^{i-1}Q_j$}. In that case, we know that an agent transits to state {\tt token} by this round. Using similar arguments as before, we can prove that agent $A$ meets an agent in state $token$ by the end of phase $i$ if it does not meet an agent in another state before. Again we get a contradiction with the assumption made at the beginning of this proof.}
\end{itemize}

{So, we get a contradiction in all cases. Hence, agent $A$ must meet some other agent at the end of phase $i$ at the latest, which proves the claim.} \finclaim

{According to Claim~2, we know that after time at most {$\Sigma_{j=1}^{i}Q_j=$\\ $\Sigma_{j=1}^{i} T(SIGN(2^j)+\Delta_{2^j}$} agent $A$ transits from state 
{\tt traveler} either to state {\tt shadow} or {\tt token} or {\tt explorer}.} {To deal with the state {\tt shadow}, we need the following claim.}

\vspace*{0.3cm}
\noindent
{{\bf Claim 3.} If agent $A$ becomes the {\tt shadow} of an agent $B$ at some round $t$, then agent $B$ cannot itself switch to state {\tt shadow} in the same round.}

{To prove the claim, there are $3$ cases to consider.}

\begin{itemize}
\item{ {Case~1.} Agent $A$ transits from state {\tt token} to state {\tt shadow} in round $t$. According to the algorithm, agent $B$ is an 
{\tt explorer} transiting to state {\tt searcher} in round $t$.}

\item{ {Case~2.} Agent $A$ transits from state {\tt searcher} to state {\tt shadow} in round $t$. According to the algorithm, $B$ is }
\begin{itemize}
\item{ either an agent in state {\tt explorer}. However, an agent in state {\tt explorer} cannot switch directly to state {\tt shadow}. Hence $B$ cannot transit to
state {\tt shadow} in round $t$. }
\item{ or an agent in state {\tt token}. Let $\mathcal{K}$ be the set of agents in state {\tt explorer} or {\tt token} that are at the same node as $A$ in round $t$. According to the algorithm, agent $B$ is the agent having the highest memory in $\mathcal{K}$. This implies that agent $B$ is not with its {\tt explorer} in round $t$ because an agent in state {\tt explorer} has a higher memory than its {\tt token} (refer to the way they are created from state {\tt traveler}). However, an agent in state {\tt token} may transit to state {\tt shadow} only if it is with its {\tt explorer}. Hence agent $B$ remains in state {\tt token} in round $t$.}
\end{itemize}
\item{ {Case~3.} Agent $A$ transits from state {\tt traveler} to state {\tt shadow} in round $t$. Let $\mathcal{J}$ be the set of the other agents that are at the same node $v$ as agent $A$ in round $t$. We have $3$ subcases to consider.}
\begin{itemize}

\item{  There are only agents in state {\tt traveler} in $\mathcal{J}$.\\ According to the algorithm, agent $B$ is a {\tt traveler} transiting to state {\tt explorer}.}

\item{  There is no agent in state {\tt explorer} or {\tt token} in $\mathcal{J}$ but at least one agent in state {\tt searcher}.\\ According to the algorithm, agent $B$ is an agent in state {\tt searcher} that does not transit to state {\tt shadow} in round $t$ (because an agent in state {\tt searcher} can transit to state {\tt shadow} only if it is with an agent in state {\tt explorer} or {\tt token}).}

\item{  There is at least one agent in state {\tt explorer} or {\tt token} in $\mathcal{J}$.\\ According to the algorithm, agent $B$ is an agent in state {\tt explorer} or {\tt token} that cannot switch to state {\tt shadow} in round $t$ for similar reasons as in Case~2. }

\end{itemize}
 
\end{itemize}

{In all cases, $B$ does not switch to state {\tt shadow} in round $t$, which proves the claim.} \finclaim

{In view of Claim~3 and of the fact that} the termination  conditions for an agent in state 
 {\tt shadow} are the same as of its guide, we may exclude the state {\tt shadow} from our analysis.

%{To continue the proof, we will need to use the following claim.}

Consider an agent in state {\tt explorer}. Either at some point it transits to state {\tt searcher}, in which case, after executing this transition,  it uses at most $\Sigma_{i=1} ^m T(EXPLO(i))$ rounds to perform procedures $EXPLO(i)$
 for $i=1,2,\dots, m$, by which time it must have met some token or explorer {(because at least one token is idle all the time starting from the first round when an agent transits to state token, according to Claim~1)} and hence must have transited to state {\tt shadow},
 or it remains in state {\tt explorer} till the end of the algorithm.
 
{ We will first show that the total number of rounds in which the agent moves as an explorer is polynomial in $m$. This is not enough to show that, after polynomial time, $A$ transits to state {\tt searcher} or remains idle forever (as an explorer), since we still need to bound the duration of each period of idleness
 between any consecutive periods of moving. This will be addressed later.}
 
{Two events can trigger further moves of agent $A$ while it is in state {\tt explorer}}: a meeting causing a non-clean exploration
 $EST^*$ or a visit of $A$ by some agent, when $A$ stays with its token after a clean exploration. 
 
 We first treat the first of these two types of events and bound the total time of explorations caused by them.
 An exploration {made by agent $A$} could be non-clean due to several reasons,
 according to the description of the algorithm.
 \begin{itemize}
 \item
{{In round $s$} agent $A$ met an agent $C$ in state {\tt token} of higher seniority than that of $A$, or of equal seniority but such
  that $recent\mbox{-}token\prec Pref_t( \cM_C)$, {where $t$ is the last round before round $s$ when the variable $recent$-$token$ of $A$ was updated.} According to the algorithm, agent $A$ transits to state {\tt searcher} as soon as it terminates its exploration $EST^*$ after round $s$. Hence such a meeting can cause at most $1$ exploration $EST^*$ of $A$ to be non-clean.}
\item
  {{In round $s$ the token $B$ of $A$ was visited by an agent $C$ in state {\tt explorer} of higher seniority than that of $B$,
     or of equal seniority but such that $Pref_t(\cM_B) \prec R$, where $R$ is the variable $recent$-$token$ of agent $C$, and $t$ is the last round before round $s$ when the variable $recent$-$token$ of $C$ was updated. According to the algorithm, agent $A$ transits to state {\tt searcher} as soon as it terminates its exploration $EST^*$ after round $s$. Hence such a meeting can cause at most $1$ exploration $EST^*$ of $A$ to be non-clean.}}

     \item
     Either agent $A$ or its token $B$ met an agent $C$ in state {\tt traveler}. Since $C$ transits immediately to state {\tt shadow}, all agents in state {\tt traveler}
    {can cause at most $m$ explorations $EST^*$ of agent $A$ to be non-clean.}
       \item
     Either agent $A$ or its token $B$ met an agent $C$ in state {\tt searcher}. Since $C$ transits immediately to state {\tt shadow}, all agents in state {\tt searcher}
   {can cause at most $m$ explorations $EST^*$ of agent $A$ to be non-clean.}
   \item
  {$A$ met an agent $C$  in state {\tt token} of lower seniority than that of $A$, or of equal seniority but such that $ Pref_k(\cM_C) \prec recent\mbox{-}token$, where $k$ is the last round before this meeting when agent $A$ updated its variable $recent\mbox{-}token$. After this meeting,
   the remaining time when agent $C$ remains in state {\tt token} is less than the longest duration of one execution of $EST^*$ (after less than this time the explorer of $C$ becomes searcher and hence $C$ transits to state {\tt shadow}). Thus, as an agent in state {\tt token}, agent $C$ can cause at most $4T(EST(m))$ explorations $EST^*$ of $A$ to be non-clean.
Hence all such meetings can cause at most $m \cdot 4T(EST(m))$ explorations $EST^*$ of $A$ to be non-clean.}

     \item
     {In round $s$} the token $B$ of $A$ met an agent $C$  in state {\tt explorer} of lower seniority than that of $B$, or of equal seniority but such that $recent\mbox{-}token \prec Pref_t(\cM_B)$, {(where $t$ is the last round before $s$ when the variable $recent$-$token$ of $C$ was updated)}.
     A similar analysis as in the previous case shows that such meetings {can cause at most $m \cdot 4T(EST(m))$ explorations $EST^*$ of $A$ to be non-clean.}  
     \item
    {In round $s$ agent $A$ met an agent $C$ in state {\tt explorer} of lower seniority than that of $A$, or of equal seniority but such that $Pref_{min(t_A,t_C)}(R_C) \prec Pref_{min(t_A,t_C)}(R_A)$, where $R_A$ (resp. $R_C$) is the value of the variable $recent$-$token$ of $A$ (resp. $C$) at the time of the meeting and $t_A$ (resp. $t_C$) is the last round before $s$ when agent $A$ (resp. $C$) updated its variable. After the meeting agent $C$ ``loses'', i.e., it will transit to state {\tt searcher} after backtracking
     to its token. Hence agent $C$ remains in state {\tt explorer} for less than $4T(EST(m))$ rounds after the meeting. Similarly as before, such meetings can cause at most $m \cdot 4T(EST(m))$  explorations $EST^*$ of $A$ to be non-clean.}

\item
    {In round $s$ agent $A$ met an agent $C$ in state {\tt explorer} of equal seniority and such that $Pref_{min(t_A,t_C)}(R_A) = Pref_{min(t_A,t_C)}(R_C)$, where $R_A$ (resp. $R_C$) is the value of the variable $recent$-$token$ of $A$ (resp. $C$) at the time of the meeting and $t_A$ (resp. $t_C$) is the last round before $s$ when agent $A$ (resp. $C$) updated its variable $recent$-$token$. Let $t'_A$ (resp. $t'_C$) be the first round since $s$ when agent $A$ (resp. $C$) updates its variable $recent$-$token$. From round $h=max(t'_A,t'_C)$ on, this kind of meeting with agent $C$ cannot occur anymore. Indeed, in view of the fact that the memories of $A$ and $C$ are necessarily different in round $s$ and the fact that once agent $C$ finishes the execution of $EST^*$ involving this meeting, its variable $recent$-$token$ will be updated, we know that  $Pref_{min(t'_A,t'_C)}(R_A) \ne Pref_{min(t'_A,t'_C)}(R_C)$ from round $h$ on. Since the difference between $max(t'_A,t'_C)$ and $s$ is less than the longest duration of one execution of $EST^*$, similarly as before, such meetings can cause at most $m \cdot 4T(EST(m))$  explorations $EST^*$ of $A$ to be non-clean.}

\item
   {In round $s$ agent $A$ met an agent $C$ in state {\tt explorer} of higher seniority than that of $A$ or of equal seniority but such that $Pref_{min(t_A,t_C)}(R_A) \prec Pref_{min(t_A,t_C)}(R_C)$, where $R_A$ (resp. $R_C$) is the value of the variable $recent$-$token$ of $A$ (resp. $C$) at the time of the meeting and $t_A$ (resp. $t_C$) is the last round before $s$ when agent $A$ (resp. $C$) updated its variable. According to the algorithm, agent $A$ transits to state {\tt searcher} as soon as it terminates its exploration $EST^*$ after round $s$. Hence such a meeting can cause at most $1$ exploration $EST^*$ of $A$ to be non-clean.}    
     
   \item
    $A$ met an agent $C$  in state {\tt token} in round $s$, that looked like its token $B$ at this time, but that turned out not to be the token $B$ after the backtrack
    of $A$ on $B$. More precisely, $recent\mbox{-}token=Pref_t(\cM_c)$ in round $s$ (where $t$ is the last round before $s$ when the variable $recent$-$token$ of $A$ was updated)
    but $Pref_s(\cM_B) \neq Pref_s(\cM_C)$. After round $s$ agent $C$ may look like token $B$ of $A$ for {less than}  $4T(EST(m))$ rounds because after {less than} this time
    $A$ backtracks to its token $B$ and, from this time on, it can see the difference between $B$ and $C$.
    Similarly as before such meetings can cause at most {$m \cdot 4T(EST(m))$  explorations $EST^*$ of $A$ to be non-clean.} 
    
    \item
    In round $t'$ the token $B$ was visited by an agent $C$ of equal seniority in state {\tt explorer} such that $recent\mbox{-}token = Pref_t(\cM_B)$, where $t$ is the last round before $t'$ when the variable $recent$-$token$ of $C$ was updated, and this agent turned out not to be $A$ after 
    the backtrack of $A$ on $B$. Similarly as before, such meetings can cause at most {$m \cdot 4T(EST(m))$ explorations $EST^*$ of $A$ to be non-clean.} 
     \end{itemize}
 {Hence the first of the two types of events (meeting causing a non-clean exploration) can cause at most {$(24T(EST(m)) +2)m+3$} explorations $EST^*$ of $A$ to be non-clean. (As before we add up all upper bounds for simplicity). Considering the fact that each non-clean exploration $EST^*$ is directly followed by at most one clean exploration $EST^*$, this kind of meeting can cause at most {$2((24T(EST(m)) +2)m+3)4T(EST(m))$} rounds of motion of $A$.}
 {The second type of events (a visit of $A$ by some agent, when $A$ stays with its token after a clean exploration)
     can cause at most $(2m(1+4T(EST(m)))+1)4T(EST(m))$ rounds of motion of $A$. Indeed, according to the algorithm, a visit of $A$ by some agent $C$ when $A$ stays idle with its token $B$
can be of the following kinds.}

\begin{itemize}
\item {Agents $A$ and $B$ are visited by an agent $C$ in state {\tt traveler}. Since, at this visit, agent $C$ transits immediately to state {\tt shadow}, such visits can trigger at most $m$ exploration $EST^*$ of agent $A$.} 

\item {Agents $A$ and $B$ are visited by an agent $C$ in state {\tt searcher}. Since, at this visit, agent $C$ transits immediately to state {\tt shadow}, such  visits can trigger at most $m$ exploration $EST^*$ of agent $A$.}

\item {In round $s$, agents $A$ and $B$ are visited by an agent $C$ in state {\tt explorer} of lower seniority than that of $B$, or of equal seniority but such that $R_C \prec Pref_{t_C}(\cM_B)$, where $\cM_B$ is the memory of agent $B$ in round $s$, $R_C$ is the variable $recent$-$token$ of agent $C$ in round $s$, and $t_C$ is the last round before $s$ when the variable $R_C$ was updated. After the visit agent $C$ ``loses'', i.e., it will transit to state {\tt searcher} after backtracking
     to its token. Hence agent $C$ remains in state {\tt explorer} for less than $4T(EST(m))$ rounds after the visit and thus, such visits can trigger at most $m \cdot 4T(EST(m))$ explorations $EST^*$ of $A$.}

\item {In round $s$ agents $A$ and $B$ are visited by an agent $C$ in state {\tt explorer} of equal seniority to that of $B$ and such that $R_C = Pref_{t_C}(\cM_B)$, where $\cM_B$ is the memory of agent $B$ in round $s$, $R_C$ is the variable $recent$-$token$ of agent $C$ in round $s$, and $t_C$ is the last round before $s$ when the variable $R_C$ was updated. Let $t'_C$ be the first round after $s$ when agent $C$ updates its variable $recent$-$token$. From round $t'_C$ on, this kind of visit by agent $C$ cannot occur anymore. Indeed, in view of the fact that the memories of $A$ and $C$ are necessarily different in round $s$ and the fact that once agent $C$ finishes the execution of $EST^*$, its variable $recent$-$token$ will be updated, we know that $R_C \ne Pref_{t'_C}(\cM_B)$ from round $t'_C$ on. Since the difference between $t'_C$ and $s$ is less than the longest duration of one execution of $EST^*$, similarly as before, such visits can trigger at most $m \cdot 4T(EST(m))$  explorations $EST^*$ of $A$.}

\item  {In round $s$, agents $A$ and $B$ are visited by an agent $C$ in state {\tt explorer} of higher seniority than that of $A$ or of equal seniority to that of $B$ and such that $Pref_{t_C}(\cM_B) \prec R_C$, where $\cM_B$ is the memory of agent $B$ in round $s$, $R_C$ is the variable $recent$-$token$ of agent $C$ in round $s$, and $t_C$ is the last round before $s$ when the variable $R$ was updated. According to the algorithm, agent $A$ transits to state {\tt searcher} as soon as it terminates its exploration $EST^*$ after round $s$. Hence such visits can trigger at most $1$ exploration $EST^*$ of A.}
\end{itemize}

{Hence adding the first exploration that must be made by $A$ (which is not trigerred by any meeting), we get an upper bound of {$(2((28T(EST(m))+3)m+3)+1) 4T(EST(m))$ rounds} during which agent $A$ moves in state {\tt explorer}.}
      
      It remains to consider an agent in state {\tt token}. It may either transit to state {\tt shadow} or remain in state {\tt token} forever.
      In the latter case it is idle all the time. 
      
      Since {{$\Sigma_{j=1}^{i} T(SIGN(2^j)+\Delta_{2^j}$}, $T(EST(m))$ and $T(EXPLO(m))$ are all polynomial in $m$}, the above
      analysis shows that there exists a polynomial $Y$, such that, for each agent $A$ executing Algorithm Gathering-without-Detection in any graph of size $m$, the 
      number of rounds during which this agent moves is at most $Y(m)$. In order to finish the proof, we need to bound the number of rounds during which an
      agent $A$ can be idle before moving again. To do this we will use the following claim.
      
      \vspace*{0.3cm}
      \noindent
      {{\bf Claim 4.}} If in round $t$ of the execution of Algorithm Gathering-without-Detection no agent moves, then no agent moves in any later round of this execution.
      
      To prove the claim notice that if no agent moves in round $t$, then in this round no agent is in state {\tt traveler} or {\tt searcher}.
      Moreover each agent in state  {\tt explorer} must be idle and stay with its token in this round (all other nodes must be in state {\tt shadow}).
      In order for some agent to move in round $t+1$, some explorer would have to visit some other token in round $t$, contradicting the definition of $t$.
      Hence all agents are idle in round $t+1$. By induction, all agents are idle from round $t$ on. This proves the claim.\finclaim
      
      %Let $Z$ be the set of rounds in an execution of Algorithm Gathering-without-Detection 
      %in which some agent is idle but will move in some subsequent round.
      %By the Claim, in each round from the set $Z$ some agent moves. 
      %Since there are at most $m$ agents in a graph of size $m$, the size of $Z$ is at most
      %$(m-1)\cdot Q(m)$. It follows that each agent must stop forever after time at most   
      %$m\cdot Q(m)$ since the wake up of the first agent.
      
      Since for each agent executing Algorithm Gathering-without-Detection in a graph of size $m$, the number of rounds in which it moves is at most $Y(m)$
      and there are at most $m$ agents, {Claim~4} implies that after time at most  $m\cdot Y(m)$ since the wake up of the first agent, all agents must stop forever.
    \end{proof} 
    
    By Lemma \ref{term2} there exists a round after which, according to Algorithm Gathering-without-Detection,  no agent moves.
    Call the resulting configuration {\em final}. The following lemma implies  that  Algorithm Gathering-without-Detection is correct.
    
    \begin{lemma}\label{cor}
    In every final configuration exactly one node is occupied by agents. 
    \end{lemma}
    
    \begin{proof}
    A final configuration must consist of agents in states  {\tt explorer} , 
{\tt token} and {\tt shadow}, all situated in nodes $v_1, \dots ,v_k$, such that in each node $v_i$ there is exactly one agent $E_i$ in state {\tt explorer}, exactly one
agent $T_i$ in state {\tt token}  and possibly some agents in state  {\tt shadow}. As before we call such a final configuration a {\em clone} 
configuration if there are at least two
distinct nodes $v_i$, $v_j$ which have identical colored views. The same argument as in the proof of Lemma \ref{clone} shows that
a final configuration cannot be a clone configuration. 

It is enough to prove that $k=1$. Suppose for contradiction that $k>1$. We will consider two cases. In the first case the memories of all explorers $E_i$ are
identical and in the second case they are not. In both cases we will derive a contradiction.

 \vspace*{0.2cm}
        \noindent
Case 1. All explorers $E_i$ in the final configuration have identical memory.

In this case all these explorers performed the last exploration $EST^*$ simultaneously, {in view of the fact that the algorithm that they execute is deterministic}.

We start with the following claim.

 \vspace*{0.2cm}
        \noindent
{\bf Claim 1.} {If a node has been rejected by the explorer $E_j$ in the construction of its {(truncated) spanning} tree during its last exploration $EST^*$, then
this node, let us call it $w$, must have been either added previously by $E_j$ to its {(truncated) spanning} tree, or added  by another explorer $E_s$
in the construction of its {(truncated) spanning} tree during its last exploration $EST^*$.} 

The node $w$ was rejected by $E_j$ for the following reason.  $E_j$ 
traveled from $w$ using the reversal $\overline{q}$ of the path $q$, where $q$ is a path (coded as a sequence of ports) from $v_j$ to some node
$u$ already in the {(truncated) spanning} tree of $E_j$, and at the end of this path $\overline{q}$, $E_j$ met a token with memory $\cM$, such that $Pref_t(\cM)=recent\mbox{-}token$,
where $t$ is the last round when $E_j$ updated its variable $recent$-$token$.

There are two possible cases. If the token met by $E_j$ is its own token (residing at $v_j$), then $w$ is equal to some node $u$ already added previously
to the {(truncated) spanning} tree of $E_j$.  If, on the other hand, the token
met by $E_j$ is the token of some other explorer $E_s$, then we will show that $w$ is added by $E_s$ to its {(truncated) spanning tree}. 
Indeed, since $E_j$ has added a node $u$
to its {(truncated) spanning} tree, such that the path from $v_j$ to $u$ is $q$, the explorer $E_s$ must have added a node $u'$ to its {(truncated) spanning} tree, such that the path from $v_s$ to $u'$ is $q$
as well, because both $E_s$ and $E_j$ have identical memories. However, this node $u'$ must be equal to $w$, since the path from $w$ to $v_s$ is $\overline{q}$.
This proves the claim.\finclaim

The contradiction in Case 1 will be obtained in the following way. Using {(truncated) spanning} trees produced by explorers $E_i$ during their last exploration $EST^*$
(recall that these trees are isomorphic, since memories of the explorers are identical), we will construct the colored view for each explorer.
Using the fact that memories of the explorers are identical, these colored views will be identical. This will imply that the final configuration is a clone
configuration, which is impossible. 

The construction proceeds as follows (we will show it for explorer $E_1$). Let $T_i$ be the {(truncated) spanning} tree produced by $E_i$.
Each tree $T_i$ has its root $v_i$ colored black and all other nodes colored white. We will gradually
attach various trees to $T_1$ in order to obtain the colored view from $v_1$. 
First attach to every node of $T_1$ its neighbors that have been rejected by $E_1$ during the construction of $T_1$. Explorer $E_1$ has visited these
nodes, hence the respective port numbers can be faithfully added. Consider any such rejected node $w$.  By Claim 1, there are two possibilities.
If node $w$ was previously added by $E_1$ to $T_1$ as some node $u$, then we proceed as follows. Let $T_1'$ be the tree $T_1$ but rooted at $u$ instead of $v_1$.
We attach tree $T_1'$ at $w$, identifying its root $u$ with $w$.
If node $w$ was added  by another explorer $E_s$ in the construction of its {(truncated) spanning} tree $T_s$, we proceed as follows. 
As mentioned in the proof of Claim 1, the explorer $E_s$ must have added a node $u'$ to its {(truncated) spanning} tree, such that the path from $v_s$ to $u'$ is $q$.
Let $T_s'$ be the tree $T_s$ but rooted at $u'$ instead of $v_s$. We attach tree $T_s'$ at $w$, identifying its root $u'$ with $w$.

After processing all nodes rejected by $E_1$ and adding the appropriate trees, we attach all rejected neighbors of nodes in the newly obtained increased
tree. These nodes could have been rejected either by $E_1$ itself or by another explorer $E_j$ whose (re-rooted) tree $T_j'$ has been attached.
For each newly attached node rejected by $E_j$, the construction continues as before, replacing the role of $T_1$ by $T_j$.

{The above construction proceeds infinitely, producing an infinite rooted tree (rooted at a node corresponding to $v_1$). We make one final addition in order to obtain the view $\cV(v_1)$ from $v_1$ (cf. Claim~2) : each node $u$ of the tree is assigned a label corresponding to the shortest path from the root to $u$. The resulting infinite tree is denoted by $\mathcal{T}$}.

\vspace*{0.3cm}
\noindent
{{\bf Claim 2.} $\mathcal{T}$ corresponds to view $\cV(v_1)$}

{We prove the claim by induction. First of all, note that $\mathcal{T}$ truncated at depth $1$ from its root corresponds to the truncated view $\cV^1(v_1)$.
Assume as induction hypothesis that the tree $\mathcal{T}$ truncated at depth $l$ from its root, call it $\mathcal{T}^l$, corresponds to the truncated view $\cV^l(v_1)$. We will show that $\mathcal{T}^{(l+1)}$ corresponds to the truncated view $\cV^{(l+1)}(v_1)$.}

{According to the process described above, whenever a node $u$, corresponding to a node $x$ in the network $\mathcal{N}$, is added to $\mathcal{T}$ under construction, exactly one child $v$ is eventually added to $u$ for each neighbor $y$ of $x$ in $\mathcal{N}$. More precisely, node $x$ is connected to node $y$ by an edge having port $p$ at node $x$ and port $q$ at node $y$ iff node $u$ and its child $v$ are connected between them by an edge having port $p$ at node $u$ and port $q$ at node $v$. In addition, if the label of $u$ is path $\rho$ then the label of $v$ is path $\rho pq$.}

{Hence, if in tree $\mathcal{T}$ we identify a node with its label, we have the following two properties: (1) for every path $\pi$ of length $2l$ (according to the induction hypothesis, $\pi$ is located at distance $l$ from the root of $\mathcal{T}$ and corresponds to a path from $v_1$ in $\mathcal{N}$), its children correspond to all paths from $v_1$ of length $2l+2$ whose prefix is $\pi$ in $\mathcal{N}$ and (2) for every path $\pi$ of length $2l$, every child $\rho$ of path $\pi$, such that $\rho=\pi pq$, is connected to $\pi$ in $\mathcal{T}$ by an edge that has port $p$ at node $\pi$ and port $q$ at node $\rho$.}

{Moreover, since $\mathcal{T}^l$ is $\cV^l(v_1)$, the above two properties also hold if $l$ is replaced by any value ranging from $0$ to $(l-1)$. Hence $\mathcal{T}^{(l+1)}$ is $\cV^{(l+1)}(v_1)$, which proves the claim.} \finclaim

%(This infinite tree is indeed the view from $v_1$ because whenever a node is added to the tree, all its neighbors are eventually added as well, 
%with the correct port numbers.)
To produce the colored view, notice that there are only two colors in this colored view: white corresponding to empty nodes in the final configuration and
black corresponding to nodes $v_1$,..., $v_k$ (all these nodes get identical colors: since memories of explorers are the same, memories of their tokens
are also the same and memories of corresponding nodes in state {\tt shadow} are also identical).  It remains to indicate how the colors are distributed
in the constructed view. This is done as follows. When a tree $T_j'$ is attached, exactly one of its nodes (namely the node corresponding to $v_j$) is black.
Exactly these nodes become black in the obtained colored view. 

This construction of colored views is done for all explorers $E_i$. Consider two explorers $E_i$ and $E_j$.
Since these explorers have the same memory, the trees $T_s'$ attached at a given stage of the construction of the views of $E_i$ and $E_j$  are
isomorphic. They are also attached in the same places of the view. Hence by induction of the level of the view it follows that both colored views are identical.
This implies that the final configuration is a clone configuration which gives a contradiction in Case 1.

 \vspace*{0.2cm}
        \noindent
Case 2.  There are at least two explorers $E_i$ and $E_j$ with different memories in the final configuration.

Consider the equivalence relation on the set of explorers $E_1, \dots ,E_k$, such that two explorers are equivalent if their memories
in the final configuration are identical.
Let $\cC_1,\dots, \cC_h$, where $h>1$, be the equivalence classes of this relation. Suppose w.l.o.g. that $\cC_1$ is a class of explorers with smallest
seniority. We will use the following claim.

\vspace*{0.3cm}
\noindent
{{\bf Claim 3.}}
During the last exploration $EST^*$ of explorers in $\cC_1$, at least one of the following statements holds:
\begin{itemize}
\item
an explorer from $\cC_1$ has visited a token of an explorer not belonging to $\cC_1$;
\item
a token of an explorer from $\cC_1$ has been visited by an explorer not belonging to $\cC_1$.
\end{itemize}

In order to prove the claim consider two cases. If every node of the graph has been visited by some explorer from $\cC_1$, we will show that the
first statement holds. Indeed, since explorers from $\cC_1$ have the smallest seniority, during their last execution of  $EST^*$ all tokens{, which are in the final configuration,} are already at 
their respective nodes {(because otherwise there would be at least one token and its explorer in the final configuration that were created after the creation of the tokens from $\cC_1$, which would be a contradiction with the fact that explorers from $\cC_1$, and hence also their tokens, have the smallest seniority in the final configuration)}. Hence some explorers from $\cC_1$ must visit the tokens of explorers outside of $\cC_1$. Hence we can restrict attention
to the second case, when some nodes of the graph have not been visited by any explorer from $\cC_1$. Notice that if there were no other classes than
$\cC_1$, this could not occur. Indeed, we would be then in Case 1 (in which all explorers have identical memory). Thus Claim 1 would hold, which implies
that all nodes must be visited by some explorer, in view of the graph connectivity. 

Hence the fact that some node is not visited by explorers from $\cC_1$ must be due to a meeting of some other agent
(which is neither an explorer from  $\cC_1$ nor a token of such an explorer) during their last exploration $EST^*$.
What kind of a meeting can it be? It cannot be a meeting with an agent in state {\tt traveler} or {\tt searcher} because this would contradict that the last exploration
was clean. For the same reason it cannot be a meeting of an explorer from $\cC_1$ with another explorer. This leaves only the two types of meetings specified in the claim, which finishes the proof of the claim.\finclaim

Let $(Ex,Tok)$ be a couple of an explorer outside of  $\cC_1$ and of its token, such that either an explorer from $\cC_1$ visited $Tok$ or a token of an explorer
from $\cC_1$ has been visited by $Ex$ during the last exploration of explorers in the class $\cC_1$. Such a couple exists by {Claim 3}. The seniority of $Ex$ and $Tok$ must be the same as that of explorers from  $\cC_1$,
for otherwise their last exploration would not be clean. For the same reason, when explorers from $\cC_1$ started their last exploration, the explorer $Ex$ must
have started an exploration as well (possibly not its final exploration):  otherwise the exploration of explorers from $\cC_1$ would not be clean. Moreover we show that
 when explorers from $\cC_1$ finished their last exploration, explorer $Ex$ must have finished an exploration as well. To prove this, consider two cases, corresponding
 to two possibilities in {Claim 3}.  Suppose that an explorer $E_j$ from $\cC_1$ has visited $Tok$ and that its exploration did not finish simultaneously with the exploration
 of $Ex$. Consider the consecutive segments $S_1,S_2,S_3,S_4$ of the last exploration $EST^*$ of $E_j$. (Recall that these segments were specified in the definition of $EST^*$.) Since $E_j$ has visited $Tok$ during $EST^*$, it must have visited it during each segment $S_i$. At the end of $S_1$, explorer $E_j$ knows how long
 $EST^*$ will take. At the end of $S_2$ its token learns it as well. When $E_j$ visits $Tok$ again in segment $S_3$, there are two possibilities. 
 Either $Tok$ does not
 know when the exploration of $Ex$ finishes, or it does know that it finishes at a different time than the exploration of $E_j$. In both cases the explorer $E_j$ that
 updated its variable $recent$-$token$ at the end of $S_2$ can see that $recent\mbox{-}token \neq Pref_t(\cM)$, where $\cM$ is the memory of $Tok$ and $t$ is the end of $S_2$. 
This makes the last exploration of $E_j$ non clean, which is a contradiction. This proves that  $Ex$ and $E_j$ finish their exploration simultaneously, if $E_j$
has visited $Tok$. The other case, when $Ex$ has visited the token of $E_j$ is similar. Hence we conclude that explorations of $E_j$ and of $Ex$ started
and finished simultaneously.

Let $\tau$ be the round in which the last exploration of $E_j$ (and hence of all explorers in $\cC_1$) finished. The exploration of $Ex$ that finished in round $\tau$
cannot be its final exploration because then it would have the same memory as $E_j$ in the final configuration and thus it would be in the class  $\cC_1$ 
contrary to the choice of $Ex$. Hence $Ex$ must move after round $\tau$. It follows that there exists a class $\cC_i$ (w.l.o.g. let it be $\cC_2$) such that explorers 
from this class started their last exploration after round $\tau$. Note that during this last exploration, explorers from $\cC_2$ could not visit all nodes of the graph, for otherwise they would meet explorers from $\cC_1$ after round $\tau$, inducing them to move after this round, contradicting the fact that explorers from $\cC_1$
do not move after round $\tau$.

The fact that some node is not visited by explorers from $\cC_2$ must be due to a meeting of some other agent
(which is neither an explorer from  $\cC_2$ nor a token of such an explorer) during their last exploration $EST^*$.
Otherwise, for explorers in $\cC_2$ the situation would be identical as if their equivalence class were the only one, and hence, as in Case 1, 
they would visit all nodes. Moreover, the fact that some node is not visited by explorers from $\cC_2$ must be due to a meeting of some explorer outside of $\cC_1$
or of its token (if not, explorers from $\cC_1$ would move after round $\tau$, which is a contradiction). 
An argument similar to that used in the proof of {Claim 3} shows that there exists a couple $(Ex',Tok')$, such that $Ex'$ is an explorer outside of  $\cC_1 \cup \cC_2$,
$Tok'$ is its token, and 
either an explorer from $\cC_2$ visited $Tok'$ or a token of an explorer
from $\cC_2$ has been visited by $Ex'$ during the last exploration of explorers in the class $\cC_2$.
Let $\tau'$ be the round in which the last exploration of explorers from $\cC_2$ is finished.
Similarly as before, the explorer $Ex'$ terminates some exploration in round $\tau'$ but continues to move afterwards.

Repeating the same argument $h-1$ times we conclude that there exists a round $\tau^*$
after which all explorers from $\cC_1 \cup \cdots \cup \cC_{h-1}$ never move again, but the last exploration of explorers from 
 $\cC_h$ starts on or after $\tau^*$. During this last exploration there must be a node not visited by any explorer from $\cC_h$, 
 otherwise some explorers from $\cC_1 \cup \cdots \cup \cC_{h-1}$ would move after $\tau^*$.
 This is due to a meeting. It cannot be a meeting with an agent in state {\tt traveler} or {\tt searcher} because this would contradict that the last exploration
was clean. For the same reason it cannot be a meeting of an explorer from $\cC_h$ with another explorer. Hence two possibilities remain.
Either an explorer from $\cC_h$ visits a token of an explorer from $\cC_1 \cup \cdots \cup \cC_{h-1}$ or a token of an explorer from $\cC_h$ is visited
by an explorer from $\cC_1 \cup \cdots \cup \cC_{h-1}$. The first situation is impossible because it would contradict the cleanliness of the last exploration
of explorers from $\cC_h$ and the second situation is impossible because explorers from $\cC_1 \cup \cdots \cup \cC_{h-1}$ do not move after $\tau^*$.
Hence in Case 2 we obtain a contradiction as well, which completes the proof.
    \end{proof}
    
    Lemmas \ref{term2} and \ref{cor} imply the following result.

\begin{theorem}
Algorithm Gathering-without-Detection performs a correct gathering of all gatherable configurations and terminates in time polynomial in the size of the graph.
\end{theorem}

\section{Consequences for leader election}

Leader election \cite{Ly} is a fundamental symmetry breaking problem in distributed
computing. Its goal is to assign, in some common round, value 1 (leader) to one of the entities and value 0 (non-leader)
to all others. The assignment should happen once for each identity, in a unique common round, and cannot be changed afterwards. 
In the context of anonymous agents in graphs, leader election can be formulated
as follows:
\begin{itemize}
\item
There exists a common unique round in which one of the agents assigns itself value 1 (i.e., it declares itself a leader) and each 
other agent assigns itself value 0 (i.e., it declares itself non-leader).
\end{itemize}

The following proposition says that the problems of leader election and of gathering with detection are equivalent in the following strong sense.
Consider any  initial configuration of agents in a graph. If gathering with detection can be accomplished for this configuration in some round $t$, then
leader election can be accomplished for this configuration in some round $t'>t$, and conversely, if leader election can be accomplished for this configuration 
in some round $t$, then gathering with detection  can be accomplished for this configuration in some round $t^*>t$.

\begin{proposition}\label{eqbis}
Leader election is equivalent to gathering with detection.
\end{proposition}

\begin{proof}
Suppose that gathering with detection is accomplished and let $t$ be the round when all agents
are together and declare that gathering is over. As mentioned in the Preliminaries, all agents must have different
memories, since they are at the same node, and, being together, they can compare these memories. 
Since,  in view of detection, the round $t$ is known to all agents,
in round $t+1$ the agent with 
the largest memory assigns itself value 1 and all other agents assign themselves value 0.

Conversely, suppose that leader election is accomplished and let $t$ be the round in which one 
of the agents assigns itself value 1 and all other agents assign themselves value 0. Starting from round $t$
the agent with value 1 stops forever and plays the role of the token, all other agents playing the role of explorers.
First, every explorer finds the token by executing procedure $EXPLO(n)$ for $T(EXPLO(n))$ rounds in phases 
$n=1,2,...$, until it finds the token in round $t'$ (this round may be different for every explorer). 
Then every explorer executes procedure $EST$ (using the token) and finds the map of the graph and hence its size $m$.
Then it waits with the token until round $t^*=t+\sum_{i=1}^m T(EXPLO(i))+T(EST(m))$. By this round
all explorers must have found the token and executed procedure $EST$, i.e., they are all together with the token.
In round $t^*$ all agents declare that gathering is over. 
\end{proof}

Proposition \ref{eqbis} implies that the class of initial configurations for which leader election is at all possible (even only using an algorithm dedicated
to this specific configuration) is equal to the class of gatherable configurations, i.e., to the class of configurations
satisfying property {\bf G}. Similarly as for gathering, we will say that a leader election algorithm is {\em universal} if it performs leader election for all
such configurations. It follows that a small modification of Algorithm Gathering-with-Detection is a universal leader election algorithm,
provided that an upper bound on the size of the graph is known to the agents. The modification is 
the following: use Algorithm Gathering-with-Detection to gather all agents in some round $t$, and, in round $t+1$, elect as leader the agent that has the largest 
memory in the round of the gathering declaration (this agent assigns itself value 1 and all other agents assign themselves value 0). Let LE be the name of this modified algorithm.

The following corollary summarizes the above discussion and gives a complete solution of the leader election
problem for anonymous agents in arbitrary graphs.

\begin{corollary}
For a given initial configuration, leader election is possible if and only if this
configuration satisfies condition {\bf G}. 
If an upper bound on the size of the graph is known, then Algorithm LE accomplishes leader election for
all these configurations.
There is no universal algorithm accomplishing leader election for all configurations satisfying condition {\bf G} in all graphs.
\end{corollary}

%\begin{acknowledgements}
%\end{acknowledgements}

%\bibliographystyle{plain}      % basic style, author-year citations
%\bibliography{dieudonne_pelc}   % name your BibTeX data base


\end{document}